\documentclass[12pt]{article}
\usepackage[english]{babel}
\usepackage[utf8]{inputenc}

\usepackage{hyperref}
\usepackage{booktabs}
\usepackage{subcaption}
\usepackage{tikz}
\usetikzlibrary{positioning}
\usetikzlibrary{arrows.meta,arrows}

\newcommand{\tbh}[1]{\textsc{\textbf{#1}}}

\title{\vspace{-2em}Technical Report: Exploring Automatic Model-Checking of the Ethereum specification\footnote{The specifications, scripts, and experimental results can be accessed at:
        \\\url{https://github.com/freespek/ssf-mc}}}
\author{Igor Konnov\thanks{This work was supported by Ethereum Foundation
    under grant FY24--1535 in the 2024 Academic Grants Round.}\\
    \footnotesize Independent Researcher\and
    Jure Kukovec\footnotemark[2] \\ \footnotesize Independent Researcher\and
    Thomas Pani\footnotemark[2] \\ \footnotesize Independent Researcher\and
    Roberto Saltini\thanks{This work was supported by Ethereum Foundation
    under grant FY24--1536 in the 2024 Academic Grants Round and was partially
    done while at Consensys.} \\
    \footnotesize Independent Researcher\and
    Thanh Hai Tran\footnotemark[3] \\ \footnotesize Independent Researcher}
\date{}

\usepackage{mathpartir}
\usepackage{amsthm}
\usepackage{amsmath}
\usepackage{amssymb}
\usepackage{mathtools}
\usepackage{listings}
\usepackage{tlalatex}

\newcommand{\SpecOne}{$\textit{Spec}_1$}
\newcommand{\SpecTwo}{$\textit{Spec}_2$}
\newcommand{\SpecThree}{$\textit{Spec}_3$}
\newcommand{\SpecThreeB}{$\textit{Spec}_{3b}$}
\newcommand{\SpecThreeC}{$\textit{Spec}_{3c}$}
\newcommand{\SpecFour}{$\textit{Spec}_4$}
\newcommand{\SpecFourB}{$\textit{Spec}_{4b}$}
\lstdefinelanguage{tla}{morekeywords={MODULE,EXTENDS,CONSTANTS,CONSTANT,ASSUME,VARIABLES,VARIABLE,
          EXCEPT,UNCHANGED,TRUE,FALSE,IF,THEN,ELSE,LET,IN,SUBSET,DOMAIN,RECURSIVE},
          comment=[l]{\\*},
  morecomment=[s]{(*}{*)},
  mathescape=true,escapechar={@},
  commentstyle=\itshape\rmfamily,keywordstyle=\sffamily\bfseries
}
\lstdefinelanguage{alloy}{morekeywords={pred,sig,fact,set,one,extends,all,or,and,run,for,but},
          comment=[l]{---},
  mathescape=true,escapechar={@},
  basicstyle=\sffamily,commentstyle=\itshape\rmfamily,keywordstyle=\sffamily\bfseries
}

\lstdefinelanguage{smt}{keywords={assert, check-sat, get-model, set, define-fun, define-fun-rec,
        not, or, and, ite, let, forall, exists, true, false
    },
    alsoletter=-,
    morekeywords={declare-const, define-sort, check-sat, set-option,
        get-assertions, get-info, get-value, set-logic,
        declare-datatypes, declare-datatype, declare-fun
    },
    sensitive=true,
    comment=[l]{;}, morecomment=[s]{(*}{*)}, string=[b]" }

\lstdefinestyle{mystyle}{
    backgroundcolor=\color{white},   
    commentstyle=\color{gray},
    keywordstyle=\color{blue},
    numberstyle=\tiny\color{gray},
    stringstyle=\color{brown},
    basicstyle=\ttfamily\footnotesize,
    breakatwhitespace=false,         
    breaklines=true,                 
    captionpos=b,                    
    keepspaces=true,                 
    numbers=left,                    
    numbersep=5pt,                  
    showspaces=false,                
    showstringspaces=false,
    showtabs=false,                  
    tabsize=2
}

\tikzset{>=latex}
\tikzset{a/.style={-{Latex[length=2mm,width=1.5mm]}}}

\newtheorem{theorem}{Theorem}[section]
\newtheorem{lemma}[theorem]{Lemma}
\newtheorem{corollary}[theorem]{Corollary}
\newtheorem{definition}{Definition}

\newcommand{\iteDef}[4]{
  #1 \coloneqq \left\{
\begin{array}{ll}
      #2 &; #3 \\
      #4 &; \text{otherwise}\\
\end{array} 
\right. 
}

\newcommand{\tlap}{$\textsc{TLA}^{+}$}

\newcommand{\nat}{\mathbb N_0}

\newcommand{\op}{\mathrm{R}}
\newcommand{\nrop}{\mathrm{I}}
\newcommand{\mop}{\mathrm{R}_m}
\newcommand{\mapg}{\mathrm{G}_m}
\newcommand{\bb}{\mathrm{next}}
\newcommand{\Chain}{\mathrm{Stack}}

\newcommand{\tup}[1]{\left<\left<#1\right>\right>}
\newcommand{\htau}{\hat{\tau}}

\newcommand{\List}{\mathrm{List}}
\newcommand{\Seq}{\mathrm{Seq}}
\newcommand{\Set}{\mathrm{Set}}
\newcommand{\PVec}{\mathrm{PVec}}
\newcommand{\PSet}{\mathrm{PSet}}
\newcommand{\PMap}{\mathrm{PMap}}
\newcommand{\Concat}{\mathrm{Concat}}
\newcommand{\Callable}{\mathrm{Callable}}
\newcommand{\Le}{\mathrm{Le}}
\newcommand{\bool}{\mathrm{bool}}
\newcommand{\Bool}{\mathrm{Bool}}
\newcommand{\pyint}{\mathrm{int}}
\newcommand{\Int}{\mathrm{Int}}
\newcommand{\ApaFoldSet}{\mathrm{ApaFoldSet}}
\newcommand{\ApaFoldSeqLeft}{\mathrm{ApaFoldSeqLeft}}
\newcommand{\MkSeq}{\mathrm{MkSeq}}
\newcommand{\SetAsFun}{\mathrm{SetAsFun}}
\newcommand{\Push}{\mathrm{Push}}
\newcommand{\At}{\mathrm{At}}
\newcommand{\Indices}{\mathrm{Indices}}
\newcommand{\chain}[0]{\mathsf{ch}}
\newcommand{\C}[0]{\mathcal{C}}

\newcommand{\recallthm}[2]{{\medskip\noindent\bfseries Theorem~\ref{#1}.~}{\itshape #2}
}

\newcommand{\recallcorollary}[2]{{\medskip\noindent\bfseries Corollary~\ref{#1}.~}{\itshape #2}
}
 
\begin{document}

\maketitle

\begin{abstract}We investigate automated model-checking of the Ethereum specification, focusing
on the \emph{Accountable Safety} property of the 3SF consensus protocol. We
select 3SF due to its relevance and the unique challenges it poses for formal
verification. Our primary tools are~\tlap{} for specification and the Apalache
model checker for verification.

Our formalization builds on the executable Python specification of 3SF\@. To
begin, we manually translate this specification into~\tlap{}, revealing
significant combinatorial complexity in the definition of Accountable Safety.
To address this, we introduce several layers of manual abstraction:
(1)~replacing recursion with folds, (2)~abstracting graphs with
integers, and (3)~decomposing chain configurations.
To cross-validate, we develop encodings in SMT (CVC5) and Alloy.

Despite the inherent complexity, our results demonstrate that exhaustive
verification of Accountable Safety is feasible for small instances ---
supporting up to 7 checkpoints and 24 validator votes. Moreover, no violations
of Accountable Safety are observed, even in slightly larger configurations.
Beyond these findings, our study highlights the importance of manual
abstraction and domain expertise in enhancing model-checking efficiency and
showcases the flexibility of~\tlap{} for managing intricate specifications.\end{abstract}\newpage 

\setcounter{tocdepth}{2}  \tableofcontents

\section{Key Outcomes}\label{sec:outcomes}

During the course of this project, we have developed a series of specifications
in \tlap{}, Alloy, and SMT, each representing a different level of abstraction
of the 3SF protocol. Before we delve into the technical details, we summarize
the key outcomes of the project:

\paragraph{Exhaustive checking of \textit{AccountableSafety}.} Our primary
objective was to verify the \textit{AccountableSafety} property of the 3SF
protocol. Model-checking this property proved to be computationally challenging
due to the unexpectedly high combinatorial complexity of the protocol.
Nonetheless, we performed systematic experiments across various specifications
in \tlap{}, Alloy, and SMT (CVC5), representing both a direct translation and
different levels of abstraction of the protocol. The largest instances we
exhaustively verified to satisfy \textit{AccountableSafety} include up to 7
checkpoints and 24 validator votes (see Table~\ref{tab:alloy-mc} in
Section~\ref{sec:alloy-results}). This comprehensive verification gives us
absolute confidence that the modeled protocol satisfies
\textit{AccountableSafety} for systems up to this size.

\paragraph{No falsification of \textit{AccountableSafety}.} In addition to the
instances where we conducted exhaustive model-checking, we ran experiments on
larger instances, which exceeded generous time limits and resulted in timeouts.
Even in these cases, no counterexamples to \textit{AccountableSafety} were
found. Furthermore, in instances where we deliberately introduced bugs into the
specifications (akin to mutation testing), Apalache, Alloy and CVC5 quickly
generated counterexamples. This increases our confidence that the protocol
remains accountably safe, even for system sizes substantially larger than
those we were able to exhaustively verify.

\paragraph{Value of producing examples.} Even though checking accountable
safety proved to be challenging, our specifications are not limited to proving
only accountable safety. They are also quite useful for producing examples. For
instance, both Apalache and Alloy are able to quickly produce examples of
configurations that contain justified and finalized checkpoints. We highlight
this unique value of specifications that are supported by model checkers:

\begin{itemize}
  \item Executable specifications in Python require the user to provide program
    inputs. These inputs can be also generated randomly, though in the case of
    3SF, this would be challenging: We expect the probability of hitting
    ``interesting'' values, such as producing finalizing checkpoints, to be
    quite low.
  \item Specifications in the languages supported by proof systems usually do
    not support model finding. The TLAPS Proof System is probably a unique
    exception here, as~\tlap{} provides a common playground for the prover and the
    model checkers~\cite{KonnovKM22}.
\end{itemize}

\paragraph{Advantages of human expertise over automated translation.} Applying
translation rules to derive checkable specifications from existing artifacts can
serve as a valuable starting point. However, such translations often introduce
inefficiencies because they cannot fully capture the nuances of the specific
context. This can result in suboptimal performance. Therefore, while
translations provide a baseline, manually crafting specifications from the
outset is usually more effective. When relying on translated specifications, it is
essential to apply manual optimizations to ensure both accuracy and efficiency.

\paragraph{Value of \tlap{}.} \tlap{} is a powerful language for specifying and
verifying distributed systems. Although our most promising experimental results
were derived from the Alloy specification, the insights gained through
iterative abstraction in \tlap{} were indispensable.\ \tlap{} enabled us to
start with an almost direct translation of the Python code and progressively
refine it into higher levels of abstraction. This iterative process provided a
deeper understanding of the protocol and laid the groundwork for the more
efficient Alloy specification. The connection between the Python specification
and the Alloy specification is definitely less obvious than the tower of
abstractions that we have built in~\tlap{}.
 
\section{Introduction}

\subsection{Quick Intro to 3SF}

At the time of writing this report, Ethereum is using Gasper~\cite{buterin2020combining} as the underlying consensus protocol.
In Gasper, time is divided into slots, which represent intervals during which a new block is proposed to extend the blockchain and undergoes voting. 
Finalizing a block -- ensuring that it is permanently added to the blockchain and cannot be reversed -- typically requires 64 to 95 slots.
This delay in finality makes the network more vulnerable to potential
block reorganizations when the network conditions change, 
e.g., during periods of asynchronous network conditions.
In particular, this finalization delay heightens the network’s exposure to
Maximal Extractable Value (MEV) exploits, 
which could undermine the network’s integrity.
Additionally, the extended finalization period forces users to weigh the
trade-off between economic security and transaction speed.

To address these issues and speed up finality, D’Amato et al.~\cite{d20243} have recently introduced the \emph{3-slot-finality} (3SF) protocols 
for Ethereum that achieve finality within three slots after a proposal, hence realizing 3-slot finality.
This feature is particularly beneficial in practical scenarios where periods of synchrony and robust honest participation often 
last much longer than the time needed for finalization in the 3SF protocol.
Finally, the 3SF protocol enhances the practicality of large-scale blockchain networks by enabling the dynamically-available component, 
which handles honest participants who may go offline and come back online~\cite{pass2017sleepy}, 
to recover from extended asynchrony, provided at least two-thirds of validators remain honest and online for sufficient time. 

To that end, the 3SF protocols combine a partially synchronous finality gadget with two dynamically available consensus protocols – 
synchronous protocols that ensure safety and liveness even with fluctuating validator participation levels. 
This design is based on the \emph{ebb-and-flow} approach introduced in~\cite{neu2021ebb}. 
An ebb-and-flow protocol comprises two sub-protocols, each with its own confirmation rule, and each outputting a chain, with one serving 
as a prefix of the other. 
The first confirmation rule defines what is known as the \emph{available chain}, which provides liveness under dynamic participation
(and synchrony). 
The second confirmation rule defines the \emph{finalized chain}, and provides safety even under network partitions, but loses liveness 
either under asynchrony or in case of fluctuation in the participation level.

\subsection{This Project: Model Checking of 3SF}

In this research project, we targeted the 3SF specification in Python\footnote{Link to the Python specification:
\href{https://github.com/saltiniroberto/ssf/blob/ad3ba2c21bc1cd554a870a6e0e4d87040558e129/high_level/common/ffg.py}{https://github.com/saltiniroberto/ssf/.../ffg.py}} as the case study focusing only on the finality gadget protocol, which is mostly specified in the file \texttt{ffg.py}.
Our main goal was to demonstrate 
\emph{Accountable Safety} of this protocol by the means of model checking. 
Accountable Safety is the property which ensures that if two conflicting chains
(i.e. chains where neither is a prefix of the other) are finalized, then -- by
having access to all messages sent -- it is possible to identify at least $\frac{1}{3}$ responsible participants.

We have chosen the specification language~\tlap{} and the model checker
Apalache for the following reasons.\ \tlap{} remains a goto language for
specifying consensus algorithms. Among the rich spectrum of
specifications~\cite{tla-examples}, the most notable for our project are the
specifications of Paxos~\cite{lamport2001paxos}, Raft~\cite{Ongaro14}, and
Tendermint~\cite{abs-1807-04938,TendermintSpec2020}. AS consensus algorithms
are quite challenging for classical model checkers like TLC, we choose
the symbolic model checker Apalache~\cite{Apalache2024,KT19,KonnovKM22}.
It utilizes the
SMT solver~Z3~\cite{MouraB08} in the background. Apalache was used for model
checking of agreement and accountable safety of
Tendermint~\cite{TendermintSpec2020}. Moreover, four of the project
participants have developed Apalache in the past and know its strenghts and
weaknesses.

\subsection{Structure of the Report}

\begin{figure}

\begin{tikzpicture}[node distance=2cm, >=latex]
    \tikzset{mynode/.style={draw, rectangle, minimum width=3cm, minimum height=1cm, align=center,
            rounded corners
        }
    }

\node[mynode, minimum width=2cm, minimum height=1cm] (py)
        {\small\texttt{ffg.py}};

    \node[mynode, minimum width=4cm,
        minimum height=1cm, below=1cm of py] (spec1)
        {\SpecOne{}: {\small\texttt{spec1-2/ffg\_recursive.tla}}};

    \node[mynode, minimum width=3cm,
        minimum height=1cm, right=1cm of spec1] (spec2)
        {\SpecTwo{}: {\small\texttt{spec1-2/ffg.tla}}};

    \node[mynode, minimum width=3cm,
        minimum height=1cm, below=1cm of spec2] (spec3)
        {\SpecThree{}: {\small\texttt{spec3/ffg.tla}}};

    \node[mynode, minimum width=3cm,
        minimum height=1cm, left=1cm of spec3] (spec4)
        {\SpecFour{}: {\small\texttt{spec4/ffg\_inductive.tla}}};

    \node[mynode, minimum width=3cm,
        minimum height=1cm, below=1cm of spec4] (spec4b)
        {\SpecFourB{}};

    \node[mynode, minimum width=3cm,
        minimum height=1cm, below left=1cm and -2cm of spec3] (spec3b)
        {\SpecThreeB{}: SMT};

    \node[mynode, minimum width=3cm,
        minimum height=1cm, below right=1cm and -2cm of spec3] (spec3c)
        {\SpecThreeC{}: Alloy/SAT };

    \draw[a] (py) -- (spec1);
    \draw[a] (spec1) -- (spec2);
    \draw[a] (spec2) -- (spec3);
    \draw[a] (spec3) -- (spec4);
    \draw[a] (spec3) -- (spec3b);
    \draw[a] (spec3) -- (spec3c);
    \draw[a] (spec4) -- (spec4b);

\end{tikzpicture}
   \caption{The relation between the specification artifacts}\label{fig:artifacts}
\end{figure}
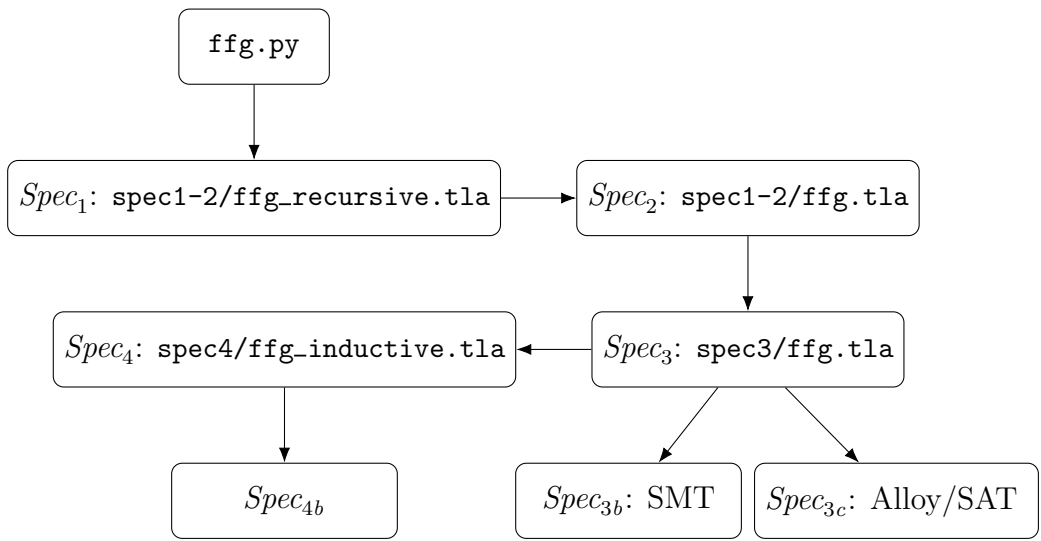

Figure~\ref{fig:artifacts} depicts the relations between the specifications
that we have produced in the project:

\begin{enumerate}
    \setcounter{enumi}{0}
    \item Section~\ref{sec:3sf} provides an introduction to the technical
    details of the 3SF protocol. We start our specification work from the
    executable specification in Python \texttt{ffg.py}.

    \item \SpecOne{}: This is the specification
        \texttt{spec1-2/ffg\_recursive.tla}. It is the result of a manual
        mechanical translation of the original executable specification in
        Python, which can be found in \texttt{ffg.py}. This specification is
        using mutually recursive operators, which are not supported by
        Apalache. As a result, we are not checking this specification. This
        specification is the result of our work in Milestones~1 and~3.
        It is discussed in Section~\ref{sec:spec1}.

    \item \SpecTwo{}: This is the specification \texttt{spec1-2/ffg.tla}. It is
        a manual adaptation of~\SpecOne{}. The main difference
        between~\SpecTwo{} and~\SpecOne{} is that~\SpecTwo{} uses ``folds''
        (also known as ``reduce'') instead of recursion. This specification is
        the result of our work in Milestones~1 and~2. It is discussed in
        Section~\ref{sec:spec2}.

    \item \SpecThree{}: This is the further abstraction of~\SpecTwo{} that uses
        the concept of a state machine, instead of a purely sequential
        specification (such as the Python code). This specification is the
        result of our work in Milestone~2. It is discussed in
        Section~\ref{sec:spec3}.

    \item \SpecFour{}: This is an extension of~\SpecThree{} that contains
        an inductive invariant in~\texttt{spec4/ffg\_inductive.tla}.
        This specification is the result of our work in Milestone~4.
        It is discussed in Section~\ref{sec:spec4}.

    \item \SpecFourB{} contains further abstractions and decomposition of
        configurations. This is the first~\tlap{} specification that allowed us
        to show accountable safety for models of very small size. This
        specification is the result of our work in Milestone~4.
        It is discussed in Section~\ref{sec:spec4b}.

    \item \SpecThreeB{} contains a specification in SMT using the theory of
        finite sets and cardinalities. In combination with the solver
        CVC5~\cite{BarbosaBBKLMMMN22}, this specification allows us to push
        verification of accountable safety even further. This specification is
        the result of our work in Milestone~4. It is discussed in
        Section~\ref{sec:smt}.

    \item \SpecThreeC{} contains a specification in
        Alloy~\cite{jackson2012software,alloytools}. With this specification,
        we manage to check all small configurations that cover the base case
        and one inductive step of the definitions. This specification is the
        result of our work in Milestone~4. It is discussed in
        Section~\ref{sec:alloy}.

    \item Appendix~\ref{section3} contains the translation rules and Appendix~\ref{proofs} contains detailed proofs
        that were conducted in Milestone~3.

\end{enumerate}

\section{Basic 3SF concepts}\label{sec:3sf}

In this section, we summarize the basic concepts of the 3SF protocol that this project depends on.
We refer the reader to the 3SF paper~\cite{d20243} for a more comprehensive explanation.

\paragraph*{Validators.} Participants of the protocol are referred to as \emph{validators} with $n$ being their total number.
Every validator is identified by a unique cryptographic identity and the public keys are common knowledge. 
Each validator has a \emph{stake} but for the purpose of this project each validator's stake is set to~1.

\paragraph*{Slots.} Time is divided into \emph{slots}. 
In ideal conditions a new block (see below) is expected to be proposed at the beginning of each slot.

\paragraph{Blocks and Chains.} 
A \emph{block} is a pair of elements, denoted as \( B = (b,p) \). Here, \( b \) represents the \emph{block body} -- essentially, the main content of the block which contains a batch of transactions grouped together.
Each block body contains a reference pointing to its \emph{parent} block. 
The second element of the pair, \( p \geq 0 \), indicates the \emph{slot} where the block \( B \) is proposed.
By definition, if $B_p$ is the parent of $B$, then $B_p.p < B.p$.
The \emph{genesis block} is the only block that does not have a parent. Its slot is \( p = 0 \).
Given the definition above, each different block $B$ implicitly identifies a different finite \emph{chain} of blocks starting from block $B$, down to the genesis block, by recursively moving from a block to its parent.
Hence, there is no real distinction between a block and the chain that it identifies.
So, by chain $\chain$, we mean the chain identified by the block $\chain$.
We write $\chain_1 \preceq \chain_2$ to mean that $\chain_1$ is a non-strict prefix of $\chain_2$.
We say that $\chain_1$ \emph{conflicts} with $\chain_2$ if and only if neither $\chain_1 \preceq \chain_2$ nor $\chain_2 \preceq \chain_1$ holds. 

\paragraph{Checkpoints.}
In the protocol described in~\cite{d20243}, a \emph{checkpoint} is a tuple $(\chain, c)$, where \(\chain \) is a chain and \( c \) is a slot signifying where \( \chain \) is proposed for justification (this concept is introduced and explained below).
However, for efficiency reasons, in the specification targeted by this project, a \emph{valid} checkpoint $\C$ is a triple $(H, c, p)$ where $H$ is the hash of a chain $\chain$, $c$ is the slot at which chain $\chain$ is proposed for justification, as per the definition above, and $p =  \chain.p$.
The total pre-order among checkpoints is defined:
$\C \leq \C'$ if and only if either \(\C.c < \C'.c\) or, in the case where \(\C.c = \C'.c\), then \(\C.p \leq \C'.p\). 
Also, $\C < \C'$ means $\C \leq \C' \land \C \neq \C'$.

\paragraph*{FFG Votes.}
Validators cast two main types of votes: \textsc{ffg-vote}s and \textsc{vote}s. 
Each \textsc{vote} includes an \textsc{ffg-vote}.
The specification targeted by this project only deals with \textsc{ffg-vote}s as the extra information included in  \textsc{vote}s has no impact on AccountableSafety.
An \textsc{ffg-vote} is represented as a tuple \([\textsc{ffg-vote}, \C_1, \C_2, v_i]\), where {$v_i$ is the validator sending the \textsc{ffg-vote}\footnote{Digital signatures are employed to ensure that $v_i$ is the actual sender and it is assumed that such digital signatures are unforgeable.}, while} \(\C_1\) and \(\C_2\) are checkpoints.
These checkpoints are respectively referred to as the \emph{source} (\(\C_1\)) and the \emph{target} (\(\C_2\)) of the \textsc{ffg-vote}.
Such an \textsc{ffg-vote} is \emph{valid}
if and only if both checkpoints are valid, 
\(\C_1.c < \C_2.c\) and  \(\C_1.\chain \preceq \C_2.\chain\). 
\textsc{ffg-vote}s effectively act as \emph{links} connecting the source and target checkpoints. Sometimes the whole \textsc{ffg-vote} is simply denoted as \(\C_1 \to \C_2\).

\paragraph*{Justification.}
A set of \textsc{ffg-vote}s is a \emph{supermajority set} if it contains valid \textsc{ffg-vote}s from at least \(\frac{2}{3}n\) distinct validators.
A checkpoint \(\C\) is considered \emph{justified} if it either corresponds to the genesis block, i.e., \(\C = (B_\text{genesis}, 0)\), or if there exists a supermajority set of links \(\{\C_i \to \C_j\}\) satisfying the following conditions. First,  for each link $\C_i \to \C_j$ in the set, {\(\C_i \to \C_j\) is valid and} \(\C_i.\chain \preceq \C.\chain \preceq \C_j.\chain\). Moreover, all source checkpoints \(\C_i\) in these links need to be already justified, and the checkpoint slot of \(\C_j\) needs to be the same as that of \(\C\) (\(\C_j.c=\C.c\)), for every \(j\). It is important to note that the source and target chain may vary across different votes.

\paragraph*{Finality.}
A checkpoint \(\C\) is \emph{finalized} if it is justified and there exists a supermajority link with source \(\C\) and potentially different targets \(\C_j\) where \(\C_j.c = \C.c + 1\). A chain \(\chain\) is finalized if there exists a finalized checkpoint \(\C\) with \(\chain = \C.\chain\). The checkpoint \(\C = (B_\text{genesis}, 0)\) is finalized by definition.

\paragraph*{Slashing.}
A validator \(v_i\) is subject to slashing for sending two \emph{distinct} \textsc{ffg-vote}s \(\C_1 \to \C_2\) and \(\C_3 \to \C_4\) if either of the following conditions holds: {\(\mathbf{E_1}\) (Double voting)} if \(\C_2.c = \C_4.c\), implying that a validator must not cast distinct \textsc{ffg-vote}s for the same checkpoint slot; or {\(\mathbf{E_2}\) (Surround voting)} if \(\C_3 < \C_1\)
 and \(\C_2.c < \C_4.c\), indicating that a validator must not vote using a lower checkpoint as source and must avoid voting within the span of its other votes.

 \begin{definition}[AccountableSafety]
  \label{def:acc-safety}
  AccountableSafety holds if and only if, upon two conflicting chains being finalized, 
 {by having access to all messages sent,} it is possible to slash at least $\frac{n}{3}$ of the validators.
\end{definition}

\subsection{Complexity of (model-checking) the protocol}

The 3SF protocol is intricate, with a high degree of combinatorial complexity,
making it challenging for automatic model checking. During our specification
work, we have observed multiple layers of complexity in the protocol:
\begin{itemize}
  \item The Python specification considers all possible graphs over all proposed
    blocks. From graph theory~\cite{cayley1878theorem}, we know that the number
    of labelled rooted forests on $n$ vertices is ${(n+1)}^{n-1}$. (Observe that
    this number grows faster than the factorial~$n!$.) This is the number of
    possible block graphs that the model checker has to consider for $n$ blocks.
  \item The protocol introduces a directed graph of checkpoints (pairs $(b,n)$
    of a block $b$ and an integer $n$) \emph{on top} of the block graph.
    Validator-signed votes form a third labeled directed graph over pairs of
    checkpoints. In addition, all of these edges have to satisfy arithmetic
    constraints.
  \item Justified and finalized checkpoints introduce an inductive structure
    that the model checker has to reason about. Essentially, the solvers have to
    reason about chains of checkpoints on top of chains of blocks.
  \item Finally, the protocol introduces set cardinalities, both for determining a
    quorum of validators and as a threshold for \textit{AccountableSafety}.
    Cardinalities are known to be a source of inefficiency in automated
    reasoning.
\end{itemize}

\section{\SpecOne{}: Translation from Python to \tlap{}}\label{sec:spec1}

\begin{figure}
    \centering
    \includegraphics[width=.9\textheight,angle=-90]{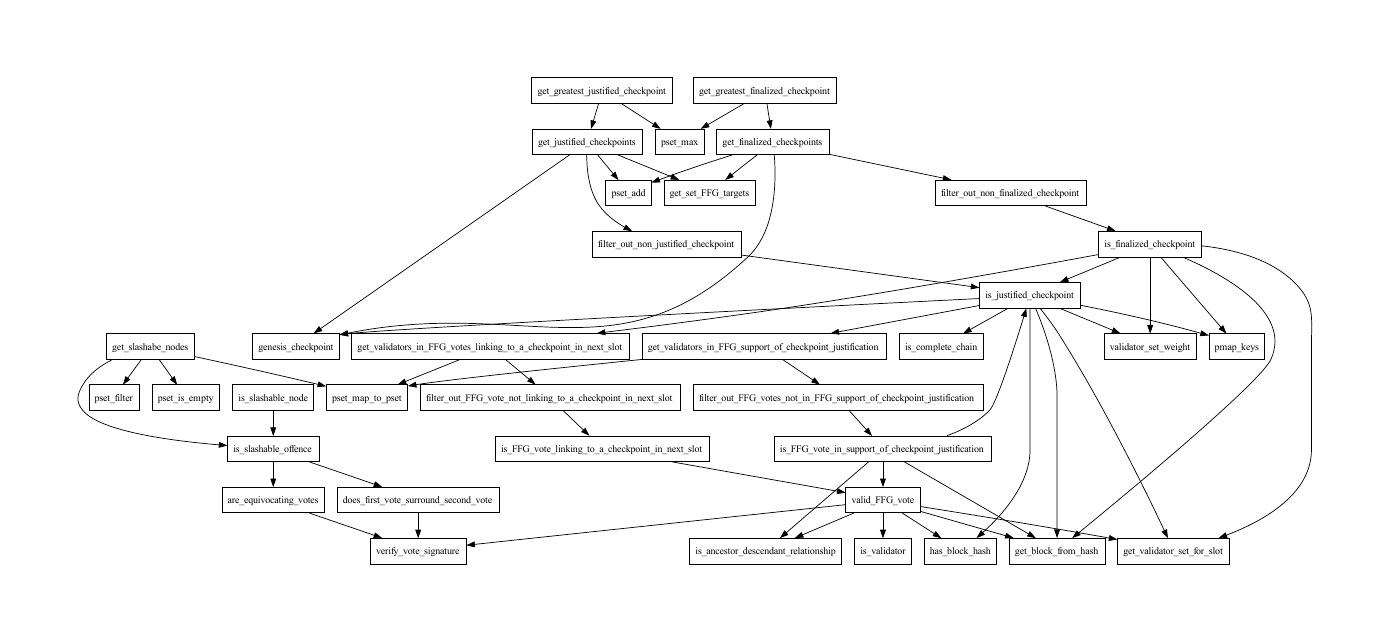}
    \caption{The callgraph of the 3SF specification Python code}
    \label{fig:callgraph}
\end{figure}

The Python specification consists of class / type definitions and a number of
functions that define the behavior of the 3SF protocol. We obtain \SpecOne{} by
performing a manual but mechanical translation of the Python specification to
\tlap{}. To this end, we first survey the Python code to identify its structure
and the program logic elements that it defines. We then translate these
functions to \tlap{} by following the principle of least surprise: we aimed to
preserve the syntax and semantics of the Python code as closely as possible --
at this stage, without regard to model-checking efficiency. We also translate
the data structures and types used in the Python code to their \tlap{}
counterparts. To this end, we have produced detailed translation rules, which we
present in Appendix~\ref{section3}.

\subsection{Shape of the Python Specification}

A few key points to note about the Python code are as follows:

\paragraph{Data structures and types.} Type definitions come in the form of
Python classes annotated with the \texttt{@dataclass} decorator, which
automatically generates constructors, comparator methods, and other boilerplate
code. Composite types are defined using the \texttt{pyrsistent} library, which
provides immutable data structures such as sets and maps.

\paragraph{Functions.} The function definitions are pure functions that
take some input and return some output -- side-effects are rare and limited to
local state. Some functions are
(mutually) recursive, as can be seen in Figure~\ref{fig:callgraph}.

\paragraph{Similarities between Python and \tlap{}.} Despite being a programming
language and specification language respectively, Python specifications written
in a functional manner (using immutable data types and pure functions) and
\tlap{} overlap somewhat in what they can express. For instance, both allow us
to express sets, structures / records, and comprehensions over these data
structures. For example, the Python code defines a record
\texttt{CommonNodeState} which we translate to a record in \tlap{}.

\subsection{Translating Basic Operators}

The Python specification defines a number of
foundational operators that are used in other functions. These operators are
mostly simple and can be translated directly to built-in operators in \tlap{}.
To take a concrete example, observe the definition of \texttt{pset\_filter} in
Figure~\ref{py_filter}.

\begin{figure}
  \begin{lstlisting}[language=Python,style=mystyle]
def pset_filter(p: Callable[[T1], bool], s: PSet[T1]) -> PSet[T1]:
  r: PSet[T1] = pset()
 
  for e in s:
    if p(e):
      r = r.add(e)
  
  return r\end{lstlisting}
\caption{\textsf{pset\_filter} definition\label{py_filter}}
\end{figure}

Notice that, while filter is not one of the built-in operators of the
\texttt{pyrsistent} library, it is relatively simple to define the
\texttt{pset\_filter} filter function, which returns a set that contains exactly
all of the elements of $s$, for which the Boolean predicate $p$ holds true.

In \tlap{}, however, filtering is a language primitive, so if we can translate a
python set $s$ to a \tlap{} set $\hat{s}$, and a Python predicate $p$ (of the
above type) to a \tlap{} predicate $\hat{p}$, we can translate
$\mathsf{pset\_filter}(p, s)$ to $\{ x \in \hat{s}\colon \hat{p}(x) \}$.

We take this idea, and apply it to every definition in the file
\texttt{pythonic\_code\_generic.py}, attempting to identify \tlap{}-equivalents
(w.r.t.\ semantics) for each defined function.  Later on, in
Appendix~\ref{section3}, we give a formal characterization of all of these
equivalencies, in the form of rewriting rules.

\subsection{Translating Complex Operators}

Some further constructs of the Python language are particularly
relevant for the translation, and we describe how we handle them in the next
paragraphs. We do, however, not give explicit rules for the translation of
\emph{all} Python language primitives to \tlap{} in general, since attempting to
establish those for the full language would vastly exceed the scope of this
project.

\paragraph{Assignments and local variables.} There are certain idiosyncrasies to
do with the fact that Python is imperative and executable, and \tlap{} is not.
For instance, Python allows for arbitrary variable assignment and reassignment,
as well as the introduction of local variables. There are two constructs
available in \tlap{} which can be used to express variable assignment:
\begin{itemize}
  \item a state-variable update $a' = e$
  \item a LET-IN local operator definition $\mathrm{LET}\; v \defeq e
    \;\mathrm{IN}\; f$
\end{itemize}

In principle, it is up to the translator to evaluate which of the two better
captures the semantics of the Python code. However, as we have noted, the Python
code is mostly functional over immutable data structures. In this project, we
have observed that we can establish a clear separation, which always translates
local/auxiliary variables to LET-IN definitions, and keeps the state-variable
updates for the state of the system.

\paragraph{Runtime exceptions.} Python code may throw at any time.
In this project in particular, the Python code is written in a defensive
manner, with many runtime checks in the form of \texttt{Requires} assertions
that perform as preconditions. 
In general, this behavior is impossible to replicate without very convoluted
\tlap{} code, so we either have to omit those assertions, or return an
unspecified value of the correct type if the requirement is not met.
We note, however, that in the present Python specification, functions are only
called if their preconditions are met (for example, the function
\texttt{get\_parent(block, node\_state)} contains a precondition
\texttt{Requires(has\_parent(block, node\_state))}, but is only ever called on
code paths where this precondition holds). Thus, we can safely omit these
assertions in the translation.

\paragraph{Control flow.} Given the functional nature of the Python
specification, and given that both languages support if-statements, those
translate directly (although since \texttt{elif} in Python has no direct
equivalent, we have to chain two IF-ELSEs in \tlap{}). Iteration (in the form of
\texttt{for}-\texttt{in} loops) only occurs in the Python code to compute the set
comprehension \texttt{pset\_filter} which we directly translate to the \tlap{}
set comprehension operator as discussed above.

\paragraph{Recursion.} As we have mentioned, the Python code contains a number
of (mutually) recursive functions. Native \tlap{} supports recursive operators
at the language level, which must be explicitly annotated with the keyword
\texttt{RECURSIVE}, which we introduce manually.

However, Apalache does not support recursive \tlap{} operators at the
model-checking level. This means that we can translate recursive Python
functions -- as written -- into \tlap{} \SpecOne{}, with the knowledge that we
will need to transform them into equivalent constructs in \SpecTwo{} to
facilitate model checking. We discuss this in more detail in the next section on
\SpecTwo{}.

\subsection{Overall Example}

As a final example, Figure~\ref{py_adr} shows a crucial operator of the Python
specification, \texttt{is\_ancestor\_descendant\_relationship}, and
Figure~\ref{tla_adr} shows its \tlap{} equivalent. It is a recursive function,
uses if-else control flow, and references \texttt{has\_parent} and
\texttt{get\_parent}, which are \tlap{}-translated python operators themselves. 

In summary, this example demonstrates that our translation strongly preserves
the syntactic shape between Python functions / compound statements and \tlap{}
operators.

\begin{figure}

  \begin{lstlisting}[style=mystyle,language=Python]
  def is_ancestor_descendant_relationship(
    ancestor: Block, 
    descendant: Block, 
    node_state: CommonNodeState
  ) -> bool:
    """
    It determines whether there is an ancestor-descendant relationship between two blocks.
    """
    if ancestor == descendant:
        return True
    elif descendant == node_state.configuration.genesis:
        return False
    else:
        return (
            has_parent(descendant, node_state) and
            is_ancestor_descendant_relationship(
                ancestor,
                get_parent(descendant, node_state),
                node_state
            )
        )\end{lstlisting}
  \caption{Python definition
      of~\textsf{is\_ancestor\_descendant\_relationship} Python}\label{py_adr}
\end{figure}

\begin{figure}
  \includegraphics[width=\textwidth]{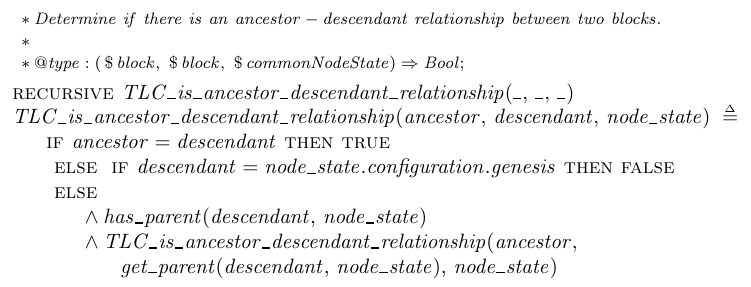}
  \caption{Definition
    of~\textsf{is\_ancestor\_descendant\_relationship} in \tlap{}}\label{tla_adr}
\end{figure}

\section{\SpecTwo{}: Fold-Based Specification in \tlap{}}\label{sec:spec2}

\SpecTwo{} addresses some of the limitations inherent in the straightforward
translation of \SpecOne{} from the Python executable specification.

\subsection{Translating Recursive \tlap{} Operators}

The primary goal of \SpecTwo{} is to maintain the semantic structure of
\SpecOne{} while eliminating recursion. In \SpecOne{}, (mutually) recursive
operators model key aspects of protocol behavior, such as the block tree and
block justification. Apalache does not natively support recursive
operators\footnote{\url{https://apalache-mc.org/docs/apalache/principles/recursive.html}},
thus it cannot be used immediately to model-check \SpecOne{}. While the
explicit-state \tlap{} model checker TLC supports recursive operators, it does
not scale to model-checking of this problem.

To resolve this, we reformulate \SpecOne{} into \SpecTwo{}, by substituting
(mutually) recursive constructs with bounded
\texttt{fold}~operations\footnote{In functional programming, \texttt{fold} (sometimes \texttt{reduce}) is a
higher-order function that accepts a combining operation and an iterable data
structure, and applies the operation to each element of the data structure
to compute a single return value.}, which enable the same iterative
computations to be performed in a non-recursive manner. We provide a set of
translation rules to convert recursive operators to bounded \texttt{fold}
operations, see Appendix~\ref{subsec:recrules}.

Let's reconsider the example in Figure~\ref{tla_adr}. In
Figure~\ref{fig:relationship-folds}, we present its equivalent that is using
folds. Following the rules in Appendix~\ref{subsec:recrules}, we have
constructed an accumulator \texttt{FindAncestor} that is passed to the fold
operator. Its body represents a single iteration of the recursive definition: it
first checks if its second parameter, indicating if an ancestor-descendant match
has already been found, is \texttt{TRUE} -- in this case, it simply propagates
the result \texttt{Pair(last\_block, TRUE)} to the next iteration (this is
equivalent to the base case of the recursive definition). If the match has not
been found, it checks if the current block is the genesis block or an orphaned
block, and returns \texttt{FALSE} if it is. Otherwise, it retrieves the parent
block of the current block, checks if it is the ancestor block, and passes
\texttt{Pair(parent, parent = ancestor)} to the next iteration. We pass this
accumulator to the fold operator, together with the initial value and the
structure to fold over (in this case, a generic sequence of length
\texttt{MAX\_SLOT}, which ensures that the fold is performed over all blocks in
the chain).

\begin{figure}
  \includegraphics[width=\textwidth]{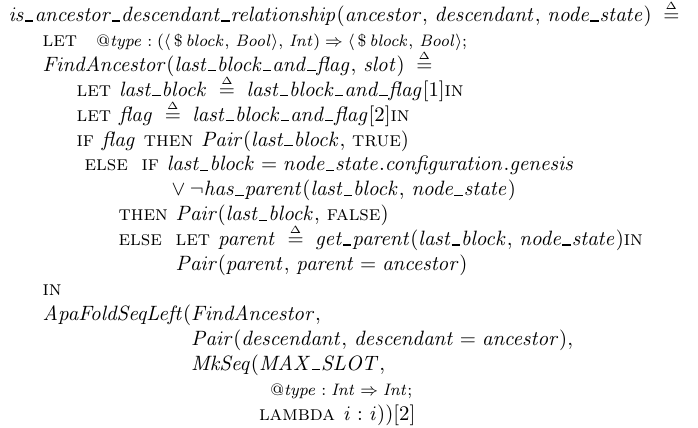}
  \caption{The recursive definition converted to bounded iteration with folds}\label{fig:relationship-folds}
\end{figure}

\subsection{An Optimization: Flattening Nested Folds}

Initial model checking experiments with \SpecTwo{} revealed significant
challenges related to memory consumption, stemming from the high number of SMT
constraints emitted by Apalache for nested fold operations, which in turn mirror
the complexity of the original nested recursive structures from the Python
specification.

To address these issues, we introduce a manual optimization strategy that
involves flattening nested fold operations. This technique transforms nested
folds into a more manageable structure by employing additional \tlap{} state
variables, similar to memoization or prophecy variables.

For example, we introduce a new \tlap{} state variable
\texttt{PRECOMPUTED\_IS\_ANCESTOR\_DESCENDANT\_RELATIONSHIP} to store
memoized ancestor-descendant relationships and initialize it with the results
of the fold operation above, see Figure~\ref{fig:relationship-memo}.

\begin{figure}
  \includegraphics[width=\textwidth]{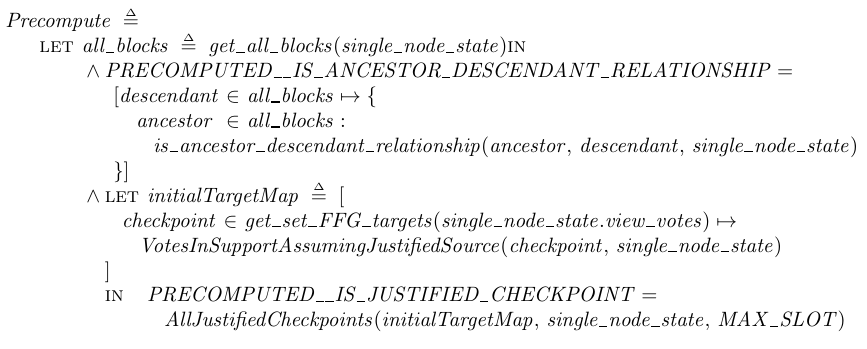}
  \caption{Memoizing the ancestor-descendant relationship}\label{fig:relationship-memo}
\end{figure}

What does this give us? For instance, instead of re-evaluating the fold
operation each time we need to check if two blocks are in an
ancestor-descendant relationship, we can directly access the memoized result in
a much more efficient map lookup, as shown in
Figure~\ref{fig:conflicting-memo}.

\begin{figure}[h]
  \includegraphics[width=\textwidth]{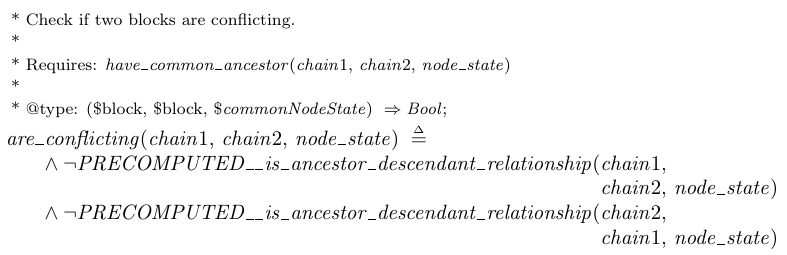}
  \caption{The definition of conflicting blocks, using memoization}\label{fig:conflicting-memo}
\end{figure}

To further improve our confidence in the correctness of this optimization, we
could produce a proof in TLAPS or run Apalache to show functional equivalence.

\subsection{Checking the Specification}

We can query the specification for reachable protocol states using Apalache.
For example, we can check if a nontrivial finalized checkpoint exists by writing an
invariant that we expect not to hold. If the invariant below is violated,
Apalache will produce an example of a finalized checkpoint as a counterexample,
see~Figure~\ref{fig:finalized-example}.

\begin{figure}[h]
  \includegraphics[width=\textwidth]{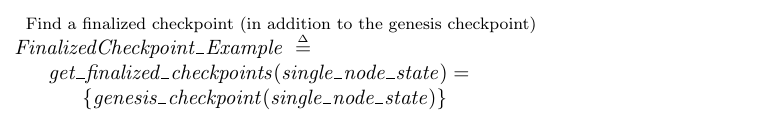}
  \caption{A false invariant for producing an example}\label{fig:finalized-example}
\end{figure}

Obviously, we can also check \textit{AccountableSafety} by supplying it as an
invariant to Apalache, see~\ref{fig:accountable-safety}.

\begin{figure}[h]
  \includegraphics[width=\textwidth]{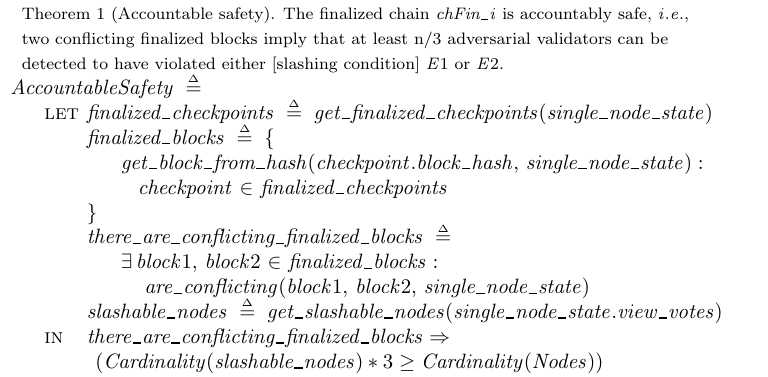}
  \caption{A state invariant for accountable safety}\label{fig:accountable-safety}
\end{figure}

Table~\ref{tab:spec2} shows the results of model checking \SpecTwo{} with
Apalache. We can see that generating examples of reachable protocol states and
verifying \textit{AccountableSafety} is infeasible due to the high computational
complexity of the specification.

\begin{table}
    \centering
    \begin{tabular}{ll}
        \tbh{Property} & \tbh{Time} \\ \toprule
        Example: conflicting blocks & timeout ($>40$h) \\
        Example: finalized \& conflicting blocks & timeout ($>40$h) \\
        AccountableSafety & timeout ($>40$h) \\ \bottomrule
    \end{tabular}
    \caption{Model checking \SpecTwo{} with Apalache.}\label{tab:spec2}
\end{table}

These results are not surprising -- the solver has to consider both reachability
properties for all possible block graphs, and all possible FFG voting scenarios
on top of these graphs.

To further evaluate \SpecTwo{}, we fix the block graph\footnote{We consider forests, rather than just trees, because the specification allows for that. Specifically, forests can occur when only some of the blocks of a given chain are received by a node. In the specification, the set of blocks received corresponds to the map \texttt{CommonNodeState.view\_blocks}.} --- this way the solver
only has to reason about voting. We encode three example block graphs: a
single, linear chain (\texttt{SingleChain}), a minimal forked chain of three
blocks (\texttt{ShortFork}), and a forest of disconnected chains
(\texttt{Forest}). Table~\ref{tab:spec2_fixed} shows the results of model
checking \SpecTwo{} for these fixed block graphs.

To understand the structure of those small graphs, see
Figure~\ref{fig:block-graphs}. In the simple
cases~\ref{fig:three}--\ref{fig:single}, we have one or two chains that form a
tree. In a more general case like in Figure~\ref{fig:forest}, a graph is a
forest. The graphs in Figures~\ref{fig:tricky1}--\ref{fig:tricky2} do not
represent chains.  The graph~\texttt{I1} is a direct-acyclic graph but not a
forest. The graph~\texttt{I2} has a loop.

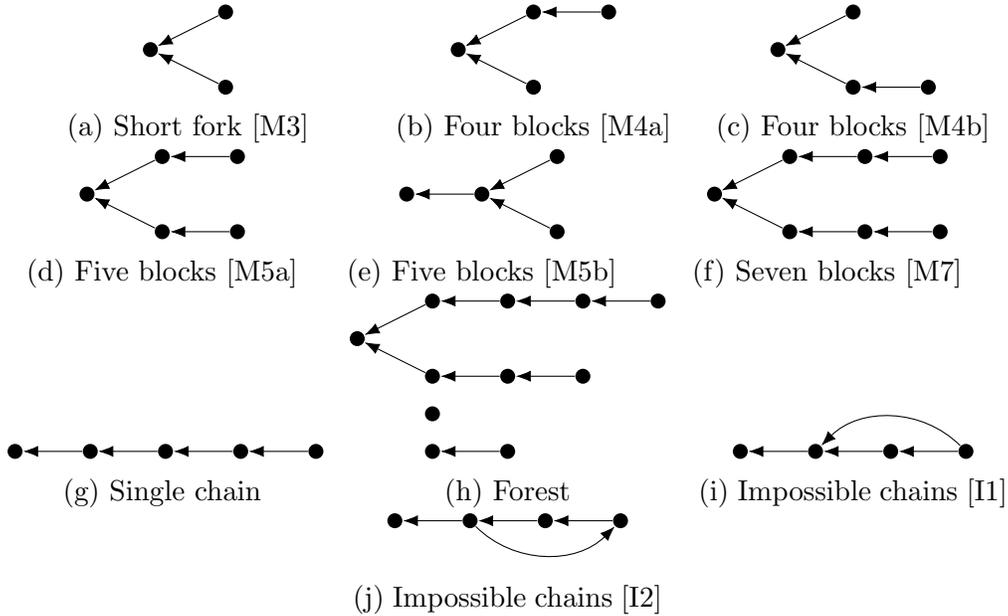
\begin{figure}

\centering
\begin{subfigure}{.35\textwidth}
  \centering
  \begin{tikzpicture}
    \tikzset{n/.style={circle, fill=black, inner sep=2pt}}

    \node[n] (n0) at (0,0) {};
    \node[n] (n1) at (1,.5) {};
    \node[n] (n1') at (1,-.5) {};
      
    \draw[a] (n1) -- (n0);
    \draw[a] (n1') -- (n0);
  \end{tikzpicture}

  \caption{Short fork [M3]}\label{fig:three}
\end{subfigure}
\begin{subfigure}{.3\textwidth}
  \centering
  \begin{tikzpicture}
    \tikzset{n/.style={circle, fill=black, inner sep=2pt}}

    \node[n] (n0) at (0,0) {};
    \node[n] (n1) at (1,.5) {};
    \node[n] (n2) at (2,.5) {};
    \node[n] (n1') at (1,-.5) {};
      
    \draw[a] (n1) -- (n0);
    \draw[a] (n2) -- (n1);
    \draw[a] (n1') -- (n0);
  \end{tikzpicture}

  \caption{Four blocks [M4a]}\label{fig:four-top}
\end{subfigure}
\begin{subfigure}{.3\textwidth}
  \centering
  \begin{tikzpicture}
    \tikzset{n/.style={circle, fill=black, inner sep=2pt}}

    \node[n] (n0) at (0,0) {};
    \node[n] (n1) at (1,.5) {};
    \node[n] (n1') at (1,-.5) {};
    \node[n] (n2') at (2,-.5) {};
      
    \draw[a] (n1) -- (n0);
    \draw[a] (n1') -- (n0);
    \draw[a] (n2') -- (n1');
  \end{tikzpicture}

  \caption{Four blocks [M4b]}\label{fig:four-bottom}
\end{subfigure}
\begin{subfigure}{.3\textwidth}
  \centering
  \begin{tikzpicture}
    \tikzset{n/.style={circle, fill=black, inner sep=2pt}}

    \node[n] (n0) at (0,0) {};
    \node[n] (n1) at (1,.5) {};
    \node[n] (n2) at (2,.5) {};
    \node[n] (n1') at (1,-.5) {};
    \node[n] (n2') at (2,-.5) {};
      
    \draw[a] (n1) -- (n0);
    \draw[a] (n2) -- (n1);
    \draw[a] (n1') -- (n0);
    \draw[a] (n2') -- (n1');
  \end{tikzpicture}

  \caption{Five blocks [M5a]}\label{fig:five1}
\end{subfigure}
\begin{subfigure}{.3\textwidth}
  \centering
  \begin{tikzpicture}
    \tikzset{n/.style={circle, fill=black, inner sep=2pt}}

    \node[n] (n0) at (0,0) {};
    \node[n] (n1) at (1,0) {};
    \node[n] (n2) at (2,.5) {};
    \node[n] (n2') at (2,-.5) {};
      
    \draw[a] (n1) -- (n0);
    \draw[a] (n2) -- (n1);
    \draw[a] (n2') -- (n1);
  \end{tikzpicture}

  \caption{Five blocks [M5b]}\label{fig:five2}
\end{subfigure}
\begin{subfigure}{.35\textwidth}
  \centering
  \begin{tikzpicture}
    \tikzset{n/.style={circle, fill=black, inner sep=2pt}}

    \node[n] (n0) at (0,0) {};
    \node[n] (n1) at (1,.5) {};
    \node[n] (n2) at (2,.5) {};
    \node[n] (n3) at (3,.5) {};
    \node[n] (n1') at (1,-.5) {};
    \node[n] (n2') at (2,-.5) {};
    \node[n] (n3') at (3,-.5) {};
      
    \draw[a] (n1) -- (n0);
    \draw[a] (n2) -- (n1);
    \draw[a] (n3) -- (n2);
    \draw[a] (n1') -- (n0);
    \draw[a] (n2') -- (n1');
    \draw[a] (n3') -- (n2');
  \end{tikzpicture}

  \caption{Seven blocks [M7]}\label{fig:seven1}
\end{subfigure}
\begin{subfigure}{.3\textwidth}
  \centering
  \begin{tikzpicture}
    \tikzset{n/.style={circle, fill=black, inner sep=2pt}}

    \node[n] (n0) at (0,0) {};
    \node[n] (n1) at (1,0) {};
    \node[n] (n2) at (2,0) {};
    \node[n] (n3) at (3,0) {};
    \node[n] (n4) at (4,0) {};
      
    \draw[a] (n1) -- (n0);
    \draw[a] (n2) -- (n1);
    \draw[a] (n3) -- (n2);
    \draw[a] (n4) -- (n3);
  \end{tikzpicture}

  \caption{Single chain}\label{fig:single}
\end{subfigure}
\begin{subfigure}{.35\textwidth}
  \centering
  \begin{tikzpicture}
    \tikzset{n/.style={circle, fill=black, inner sep=2pt}}

    \node[n] (n0) at (0,0) {};
    \node[n] (n1) at (1,0.5) {};
    \node[n] (n2) at (2,0.5) {};
    \node[n] (n3) at (3,0.5) {};
    \node[n] (n4) at (4,0.5) {};
    \node[n] (n5) at (1,-.5) {};
    \node[n] (n6) at (2,-.5) {};
    \node[n] (n7) at (3,-.5) {};
    \node[n] (n8) at (1,-1) {};
    \node[n] (n9) at (1,-1.5) {};
    \node[n] (n10) at (2,-1.5) {};
      
    \draw[a] (n1) -- (n0);
    \draw[a] (n2) -- (n1);
    \draw[a] (n3) -- (n2);
    \draw[a] (n4) -- (n3);
    \draw[a] (n5) -- (n0);
    \draw[a] (n6) -- (n5);
    \draw[a] (n7) -- (n6);

    \draw[a] (n10) -- (n9);
  \end{tikzpicture}

  \caption{Forest}\label{fig:forest}
\end{subfigure}
\begin{subfigure}{.3\textwidth}
  \centering
  \begin{tikzpicture}
    \tikzset{n/.style={circle, fill=black, inner sep=2pt}}

    \node[n] (n0) at (0,0) {};
    \node[n] (n1) at (1,0) {};
    \node[n] (n2) at (2,0) {};
    \node[n] (n3) at (3,0) {};
      
    \draw[a] (n1) -- (n0);
    \draw[a] (n2) -- (n1);
    \draw[a] (n3) -- (n2);
    \draw[a,bend right=45] (n3) to (n1);
  \end{tikzpicture}

  \caption{Impossible chains [I1]}\label{fig:tricky1}
\end{subfigure}
\begin{subfigure}{.3\textwidth}
  \centering
  \begin{tikzpicture}
    \tikzset{n/.style={circle, fill=black, inner sep=2pt}}

    \node[n] (n0) at (0,0) {};
    \node[n] (n1) at (1,0) {};
    \node[n] (n2) at (2,0) {};
    \node[n] (n3) at (3,0) {};
      
    \draw[a] (n1) -- (n0);
    \draw[a] (n2) -- (n1);
    \draw[a] (n3) -- (n2);
    \draw[a,bend right=45] (n1) to (n3);
  \end{tikzpicture}

  \caption{Impossible chains [I2]}\label{fig:tricky2}
\end{subfigure}

   \caption{Small instances of chains and non-chains}\label{fig:block-graphs}
\end{figure}

\begin{table}
    \centering
    \begin{tabular}{llr}
      \tbh{Property} & \tbh{Block graph} & \tbh{Time} \\ \toprule
      Example: conflicting blocks & \texttt{SingleChain} & 1 min 3 sec \\
      Example: conflicting blocks & \texttt{ShortFork} & 52 sec \\
      Example: conflicting blocks & \texttt{Forest} & 2 min 21 sec \\ \midrule
      Example: fin.\ \& confl.\ blocks & \texttt{SingleChain} & 1 min 5
      sec \\
      Example: fin.\ \& confl.\ blocks & \texttt{ShortFork} & 10 hours
      49 min 47 sec \\
      Example: fin.\ \& confl.\ blocks & \texttt{Forest} & timeout
      ($>40$h) \\ \midrule
      AccountableSafety & \texttt{SingleChain} & 1 min 13 sec \\
      AccountableSafety & \texttt{ShortFork} & timeout ($>40$h) \\
      AccountableSafety & \texttt{Forest} & timeout ($>40$h) \\ \bottomrule
    \end{tabular}
    \caption{Model checking \SpecTwo{} for fixed block
    graphs.}\label{tab:spec2_fixed}
\end{table}

We can see that the solver can handle the single chain block graph (where
\textit{AccountableSafety} trivially holds due to absence of conflicting
blocks), but struggles with the more complex scenarios even when given a fixed
block graph. This suggests that the complexity inherent in the specification is
due to the high combinatorial complexity of voting scenarios, rather than just
the block graph.

\section{\SpecThree{}: Adding a State Machine}\label{sec:spec3}

In the course of writing~\SpecTwo{}, we realized
that the executable Python specification is essentially sequential. In other
words, even though the 3SF algorithm is distributed, its Python
specification as well as~\SpecOne{} and~\SpecTwo{} are encoding
all possible protocol states in a single specification state.

\subsection{State Machines in \tlap{}}

Since \tlap{} is designed for reasoning about state machines, Apalache is tuned
towards incremental model checking of the executions. For instance, if a state
machine is composed of $n$ kinds of state-transitions (called actions), that
is, $\mathit{Next} = A_1 \vee \dots \vee A_N$, the model checker tries to find
a violation to a state invariant~$\textit{Inv}$ by assuming that a single
symbolic transition $A_i$ took place. If there is no violation, that instance
of the invariant~$\textit{Inv}$ can be discarded. By doing so, the model
checker reduces the number of constraints for the SMT solver to process.  The
same applies to checking an inductive invariant. When there is no such
decomposition of $\mathit{Next}$, the model checker produces harder problems
for SMT\@.

\paragraph{Symbolic simulation.} In addition to model checking, Apalache offers
support for symbolic simulation. In this mode, the model checker
non-deterministically selects \emph{one} symbolic transition at each step and
applies it to the current state. While symbolic simulation is not exhaustive
like model checking, it allows for more efficient exploration of deep system
states. Importantly, any counterexample discovered through symbolic simulation
is just as valid as one found via exhaustive model checking.

\subsection{Introducing a State Machine}

Having the observations above in mind, we introduce~\SpecThree{} in which we
specify a state machine that incrementally builds possible protocol states by
constructing the following data structures:

\begin{itemize}

    \item The set of proposed blocks, and the graph containing these blocks,
        called $\textit{blocks}$ and $\textit{graph}$, respectively.

    \item The ancestor-descendant relation, called
        $\textit{block\_graph\_closure}$.

    \item The announced FFG votes and the validators' votes on them, called
        $\textit{ffg\_votes}$ and $\textit{votes}$, respectively.

    \item The set of justified checkpoints that is computed as the greatest
        fixpoint, called $\textit{justified\_checkpoints}$.

\end{itemize}

The most essential part of the specification is shown in
Figure~\ref{fig:abstract-ffg-cast-votes}. It represents an abstract
state-transition that corresponds to a subset of validators sending votes for
the same FFG vote. The most interesting part of this transition can be seen in
the last lines: Instead of directly computing the set of justified checkpoints
in the current state, we just ``guess'' it and impose the required constraints
on this set. Mathematically, $\textit{justified\_checkpoints}$ is the greatest
fixpoint among the checkpoints that satisfy the predicate
$\textit{IsJustified}$.

\begin{figure}
    \centering
    \includegraphics[width=\textwidth]{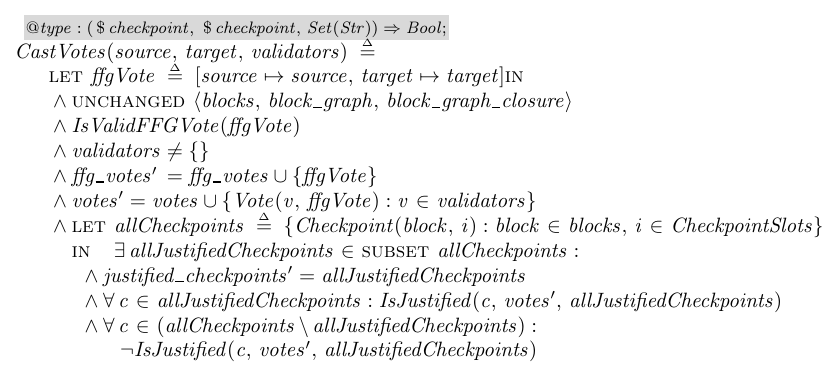}  \caption{A state-transition that casts votes}\label{fig:abstract-ffg-cast-votes}
\end{figure}

Figure~\ref{fig:abstract-ffg-next} shows the transition predicate of the
specification. It just non-deterministically chooses the inputs to
$\textit{ProposeBlock}$ or $\textit{CastVotes}$ and fires one of those two
actions.

\begin{figure}
    \centering
    \includegraphics[width=\textwidth]{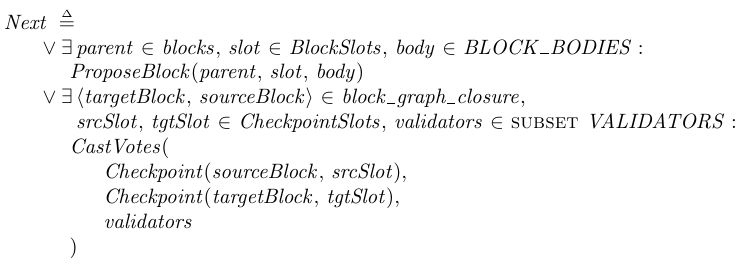}  \caption{The transition predicate of~\SpecThree{}}\label{fig:abstract-ffg-next}
\end{figure}

\subsection{Model Checking Experiments}

We have conducted model checking experiments with this specification. They are
summarized in Table~\ref{tab:abstract-ffg-mc}. Two things to notice:

\begin{itemize}

  \item The model checker finds examples of justified and conflicting blocks
    very fast.

  \item Model checking accountable safety gets stuck very fast, just after 6
    steps. Note that justified blocks are recomputed after every step that
    casts votes.

\end{itemize}

\begin{table}
    \centering
    \begin{tabular}{llrr}
        \tbh{Command}
            & \tbh{State invariant}
            & \tbh{Depth}
            & \tbh{Time}
            \\ \toprule
        \texttt{check}
            & $\textit{ExistsJustifiedNonGenesisInv}$
            & 2
            & 5s
            \\ \midrule
        \texttt{check}
            & $\textit{ExistTwoFinalizedConflictingBlocks}$
            & 6
            & 52min
            \\ \midrule
        \texttt{simulate}
            & $\textit{AccountableSafety}$
            & 6
            & 37min
            \\ \midrule
        \texttt{check}
            & $\textit{AccountableSafety}$
            & 6
            & 16h 13min
            \\ \bottomrule
    \end{tabular}
    \caption{Model checking experiments
             with~\SpecThree{}}\label{tab:abstract-ffg-mc}
\end{table}

\section{\SpecThreeB{} in SMT}\label{sec:smt}

In addition to \SpecThree{}, we have developed a manual encoding of the 3SF
protocol in SMT, directly utilizing the CVC5 solver~\cite{BarbosaBBKLMMMN22}.
In contrast to the \tlap{}-based specifications, we manually encode an
arbitrary state of the protocol
directly, using SMT-LIB constraints. At the cost of using a lower-level language and
requiring a specialized solver, this has the following advantages:
\begin{itemize}
  \item The manual encoding is more succinct than the SMT encoding produced by
    Apalache from \tlap{}.
  \item SMT-LIB supports recursive functions, which allows us to express the
    recursion inherent to the 3SF protocol more naturally.
\end{itemize}

Structurally, this specification combines some aspects of \SpecTwo{} and
\SpecThree{}: like \SpecTwo{}, it does not encode a state machine, but rather
an arbitrary state of the protocol. However, it uses data structures and
constraints similar to \SpecThree{} to model key components of the 3SF
protocol, including checkpoints, FFG votes, justification and finalization,
and slashing conditions.
To model finite sets and cardinalities, we use the non-standard SMT theory of
sets and cardinalities~\cite{DBLP:journals/lmcs/BansalBRT18} provided by CVC5.

\subsection{Modeling}
The SMT spec explicitly introduces hashes, checkpoints and nodes as atoms over
finite domains. Votes are modeled as any possible combination
of a source and target checkpoint and a sending node:

\begin{lstlisting}[language=smt]
(declare-datatype Hash ((Hash1) (Hash2) (Hash3)))
(declare-datatype Checkpoint ((C1) (C2) (C3) (C4) (C5)))
(declare-datatype Node ((Alice) (Bob) (Charlie) (David)))
(declare-datatype Vote ((Vote (source Checkpoint) (target Checkpoint) (sender Node))))
\end{lstlisting}

To remain within the decidable SMT fragment, we have to model unbounded data
using functions. For example, we model the slot number of a block as a function
from block hashes to integers:

\begin{lstlisting}[language=smt]
(declare-fun slot (Hash) Int)
; genesis' slot is 0
(assert (= (slot genesis) 0))
; slots are increasing from parent to child
(assert (forall ((h Hash)) (=> (not (= h genesis)) (> (slot h) (slot (parent_of h))))))
\end{lstlisting}

We encode all protocol rules as declarative constraints in the SMT model. For
example, the constraint that a checkpoint is justified if and only if there is
a supermajority of validators that cast a justifying vote from an already
justified checkpoint is encoded as follows:

\begin{lstlisting}[language=smt]
(declare-const justified_checkpoints (Set Checkpoint))
(assert (= justified_checkpoints (set.comprehension ((c Checkpoint))
  (or
    ; L3: genesis is justified
    (= c genesis_checkpoint)
    ; L4: there is a quorum of validators that cast a vote from a justified checkpoint to c
    (>= (* 3 (set.card (set.comprehension ((node Node))
        (exists ((vote Vote)) (and
          ; L4+5: vote is a valid vote cast by node
          (set.member vote votes)
          (= (sender vote) node)
          ; L6: the source of the vote is justified
          (set.member (source vote) justified_checkpoints)
          ; L7: there is a chain source.block ->* checkpoint.block ->* target.block
          (and
            (set.member
              (tuple (checkpoint_block (source vote)) (checkpoint_block c))
              ancestor_descendant_relationship)
            (set.member
              (tuple (checkpoint_block c) (checkpoint_block (target vote)))
              ancestor_descendant_relationship)
          )
          ; L8: the target checkpoint slot is the same as the checkpoint's
          (= (checkpoint_slot (target vote)) (checkpoint_slot c))
        ))
        node
      )))
      (* 2 N)
    )
  )
  c
)))
\end{lstlisting}

\subsection{Checking the Specification}

Similar to model checking with Apalache, we can use the CVC5 solver to find
examples of reachable protocol states, often within seconds. For example, to
find a finalized checkpoint besides the genesis checkpoint, we append the
following SMT-LIB script to our specification:

\begin{lstlisting}[language=smt]
; find a finalized checkpoint (besides genesis)
(assert (not (= finalized_checkpoints (set.singleton genesis_checkpoint))))
(check-sat)
(get-model)
\end{lstlisting}

Ultimately, we want to show that \textit{AccountableSafety} holds. To check that
this property holds with an SMT solver, we negate the property and check for
unsatisfiability:

\begin{lstlisting}[language=smt]
; there is a counterexample to AccountableSafety
(assert (and
  (exists ((block1 Hash) (block2 Hash))
    (and
      (set.member (tuple block1 block2) conflicting_blocks)
      (set.member block1 finalized_blocks)
      (set.member block2 finalized_blocks)
  ) )
  (< (* 3 (set.card slashable_nodes)) N)
))
(check-sat)
\end{lstlisting}

If the constraint above was satisfiable, we would have found a counterexample
to \textit{AccountableSafety}. Thus, \textit{AccountableSafety} holds if the
solver returns unsatisfiable.

\subsection{Experimental Results}

Similar to our experiments with Apalache, we instantiate the SMT model with
different numbers of blocks and checkpoints to verify
\textit{AccountableSafety} and evaluate the runtime of the solver.
Table~\ref{tab:smt} shows the results of our experiments.

\begin{table}
    \centering
    \begin{tabular}{rrrrrrrr}
        \tbh{blk} & \tbh{chk} & \tbh{Time} \\ \toprule
        3 & 5 & 8 min 11 sec \\
        4 & 5 & 22 min 00 sec \\
        5 & 5 & 1 h 40 min 19 sec \\ \midrule
        3 & 6 & 74 h 1 min 41 sec \\
        4 & 6 & timeout ($>80$ h) \\
        5 & 6 & timeout ($>80$ h) \\ \bottomrule
    \end{tabular}
    \caption{Model checking experiments with CVC5.\ \textbf{blk} is the number of blocks, \textbf{chk} is the number
      of checkpoints.}\label{tab:smt}
\end{table}

As we can see, similar to our \tlap{} experiments, the runtime of the SMT solver
grows exponentially with the number of blocks and checkpoints. The solver is
able to verify \textit{AccountableSafety} for up to 5 blocks and 5 checkpoints
or 3 blocks and 6 checkpoints, but it times out for larger instances.

\section{\SpecThreeC{} in Alloy}\label{sec:alloy}

Once we saw that our specifications~\SpecThree{} and~\SpecFour{}
were too challenging for the \tlap{} tools, we decided to employ
Alloy~\cite{jackson2012software,alloytools}, as an alternative to our tooling.
Alloy has two features that are attractive for our project:

\begin{itemize}

  \item Alloy allows to precisely control the search scope by setting the
      number of objects of particular type in the universe. For example, we can
      restrict the number of checkpoints and votes to 5 and 12, respectively.
      Although we introduced similar restrictions with Apalache, Alloy has even
      finer level of parameter tuning.

  \item Alloy translates the model checking problem to a Boolean satisfiability
      problem (SAT), in contrast to Apalache, which translates it to satisfiability-modulo-theory (SMT).  This
      allows us to run the latest off-the-shelf solvers such as Kissat, the
      winner of the SAT Competition 2024~\cite{SAT-Competition-2024-solvers}.

\end{itemize}

\subsection{From \tlap{} to Alloy}

Having~\SpecThree{} in~\tlap{}, it was relatively easy for us to write down
the Alloy specification, as~\SpecThree{} is already quite abstract. However, as
Alloy's core abstractions are objects and relations between them, the
specification looks quite differently. For example, here is how we declare
signatures for the blocks and checkpoints:

\begin{lstlisting}[language=alloy,columns=fullflexible]
  sig Payload {}
  sig Signature {}
  fact atLeastFourSignatures { #Signature >= 4 }

  sig Block {
    slot: Int,
    body: Payload,
    parent: Block
  }
  sig Checkpoint {
    block: Block,
    slot: Int
  }

  one sig GenesisPayload extends Payload {}
  one sig GenesisBlock extends Block {}
\end{lstlisting}

Unlike~\SpecThree{} but similar to~\SpecThreeB{}, our Alloy specification does
not describe a state machine, but it specifies an arbitrary single state of the
protocol. As in~\SpecThree{}, we compute the set of justified checkpoints as a
fixpoint:

\begin{lstlisting}[language=alloy,columns=fullflexible]
  one sig JustifiedCheckpoints {
    justified: set Checkpoint
  }
  fact justifiedCheckpointsAreJustified {
    all c: JustifiedCheckpoints.justified |
      c.slot = 0 or 3.mul[#justifyingVotes[c].validator] >= 2.mul[#Signature]
  }
\end{lstlisting}

\subsection{Model Checking with Alloy}

Similar to our experiments with Apalache, we produce examples of configurations
that satisfy simple properties. For example:

\begin{lstlisting}[language=alloy,columns=fullflexible]
  run someFinalizedCheckpoint { some c: Checkpoint |
    isFinalized[c] and c.slot != 0
  }
  for 10 but 6 Block, 6 Checkpoint, 12 Vote, 5 int
\end{lstlisting}

For small search scopes, Alloy quickly finds examples, often in a matter
of seconds.

Ultimately, we are interested in showing that~$\textit{AccountableSafety}$
holds true. To this end, we want Alloy to show that the
formula~$\textit{noAccountableSafety}$ does not have a model (unsatisfiable).
Similar to Apalache, Alloy allows us to show safety for bounded scopes.

\begin{lstlisting}[language=alloy,columns=fullflexible]
  pred noAccountableSafety {
    disagreement and (3.mul[#slashableNodes] < #Signature)
  }
  run noAccountableSafety for 10 but 6 Block, 6 Checkpoint,
                                     4 Signature, 6 FfgVote, 15 Vote, 5 int
\end{lstlisting}

Thus it is interesting to increase these parameter values as much as possible
to cover more scenarios, while keeping them small enough that the SAT solver
still terminates within reasonable time. Table~\ref{tab:alloy-mc} summarizes
our experiments with Alloy when checking $\textit{noAccountableSafety}$ with
Kissat.  Since a vote includes an FFG vote and a validator's signature, the
number of FFG votes and validators bound the number of votes. Hence, if we only
increase the number of votes without increasing the number of FFG votes, we may
miss the cases when longer chains of justifications must be constructed.

\subsection{Model Checking Results}\label{sec:alloy-results}

As we can see, we have managed to show \textit{AccountableSafety} for chain
configurations of up to 6 blocks, including the genesis block. While the
configurations up to 4 blocks are important to analyze, they do not capture the
general case, namely, of two finalized checkpoints building on top of justified
checkpoints that are not the genesis. Thus, configurations of 5 or more blocks
are the most important ones, as they capture the general case.

To see if we could go a bit further beyond 5 blocks, we have tried a
configuration with 7 blocks. This configuration happened to be quite challenging
for Alloy and Kissat. We have run Kissat for over 16 days, and it was stuck at
the remaining seven percent for many days. At that point, it had introduced:

\begin{itemize}
  \item 60 thousand Boolean variables,
  \item 582 thousand ``irredundant'' clauses and 199 thousand binary clauses,
  \item 1.3 billion conflicts,
  \item 56 million restarts.
\end{itemize}

We conjecture that the inductive structure of justified and finalized
checkpoints make it challenging for the solvers (both SAT and SMT) to analyze
the specification. Essentially, the solvers have to reason about chains of
checkpoints on top of chains of blocks.

\begin{table}
    \centering
    \begin{tabular}{rrrrrrrr}
        \tbh{\#}
            & \tbh{blk}
            & \tbh{chk}
            & \tbh{sig}
            & \tbh{ffg}
            & \tbh{votes}
            & \tbh{Time}
            & \tbh{Memory}
            \\ \toprule
        1 & 3 & 5 & 4 & 5  & 12 & 4s & 35M
            \\
        2 & 4 & 5 & 4 & 5 & 12 & 10s & 40M
            \\
        3 & 5 & 5 & 4 & 5 & 12 & 15s & 45M
            \\
        4 & 3 & 6 & 4 & 6 & 15 & 57s & 52M
            \\
        5 & 4 & 6 & 4 & 6 & 15 & 167s & 55M
            \\
        6 & 5 & 6 & 4 & 6 & 15 & 245s & 57M
            \\
        7 & 6 & 6 & 4 & 6 & 15 & 360s & 82M
            \\
        8 & 5 & 7 & 4 & 6 & 24 & 1h 27m & 156M
            \\
        9 & 5 & 10 & 4 & 8 & 24 & over 8 days (timeout) & 198 MB
            \\
        9 & 5 & 10 & 4 & 8 & 32 & over 8 days (timeout) & 220 MB
            \\
        10 & 3 & 15 & 4 & 5 & 12 & 31s & 56M
            \\
        11 & 4 & 20 & 4 & 5 & 12 & 152s & 94M
            \\
        12 & 5 & 25 & 4 & 5 & 12 & 234s & 117M
            \\
        13 & 7 & 15 & 4 & 10 & 40 & over 16 days (timeout) & 300 MB
            \\ \bottomrule
    \end{tabular}
    \caption{Model checking experiments with Alloy and Kissat:
      \textbf{blk} is the number of blocks, \textbf{chk} is the number
      of checkpoints, \textbf{sig} is the number of validator signatures,
      \textbf{ffg} is the number of FFG votes, \textbf{votes} is the number
      of validator votes.
    }\label{tab:alloy-mc}
\end{table}

\section{\SpecFour{}: Two Chains in \tlap{}}\label{sec:spec4}

We obtain this specification by preserving the vote and checkpoint behavior from \SpecThree{}, but we restrict the block-graph to exactly two chains, which are allowed to fork, i.e., they always share a common prefix.
To this end, we restrict block bodies of the two chains as follows:
\begin{enumerate}
	\item Block bodies on the first chain have non-negative sequential integer numbers starting from 0 (genesis), e.g. $0, 1, 2,3, 4$.
	\item Block bodies on the second chain are the same as the ones on the first chain \emph{up to the fork point}, after which they change sign, but otherwise maintain their sequence in absolute value terms, e.g. $0, 1,2,-3,-4$.
\end{enumerate}

The aim of this encoding is to decrease complexity for the SMT solver, since it simplifies block graph ancestry and closure reasoning significantly (e.g. we can check whether one block is an ancestor of another by simply comparing block body integer values), while preserving the behavior which arises from voting on checkpoint justification.

\subsection{Inductive Invariant}\label{sec:spec4-indinv}

This specification defines an inductive invariant \textit{IndInv}. Recall, an inductive invariant satisfies the following two conditions, assuming we have an initial-state predicate \textit{Init}, and transition predicate \textit{Next}:
\begin{enumerate}
	\item It is implied by the initial state: $\mathit{Init} \Rightarrow \mathit{IndInv}$
	\item It is preserved by the transition predicate: $\mathit{IndInv} \land \mathit{Next} \Rightarrow \mathit{IndInv}$
\end{enumerate}
Because of this characterization, inductive invariants lend themselves especially nicely to bounded symbolic model checking; with Apalache, one can prove (or disprove) an inductive invariant by running two queries of depth at most 1, corresponding to the above properties.
If no violation is found, we are assured that \textit{IndInv} holds in all possible reachable states.
Since our goal is ultimately to prove or disprove \textit{AccountableSafety}, we can additionally prove $\mathit{IndInv} \Rightarrow \mathit{AccountableSafety}$.

The challenge, typically, is that inductive invariants are more difficult to write, compared to state invariants.
They are usually composed of several lemmas, i.e. properties that we are less interested in on their own, but which are crucial in establishing property (2.).

As we explain below, the inductive invariant introduced in \SpecFour{} mainly consists of two sets of lemmas, one for characterizing the chain-fork scenarios, and a second one for characterizing justified checkpoints, shown in Figure~\ref{figFork} and Figure~\ref{figJC} respectively.

\begin{figure}
  \includegraphics[width=\textwidth]{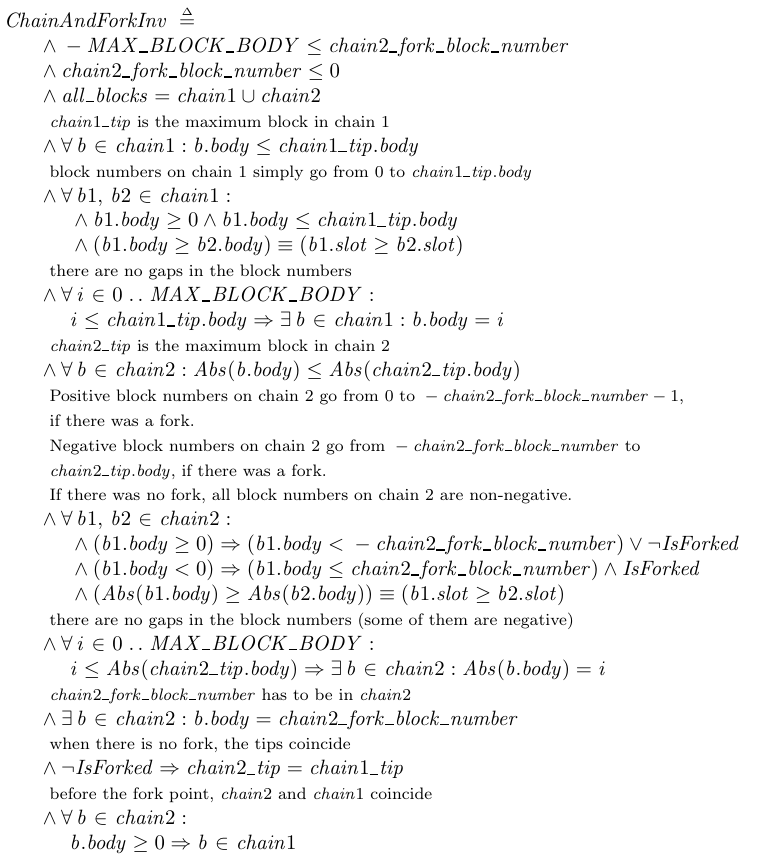}
  \caption{Chain and forking lemmas in \textit{IndInv}}\label{figFork}
\end{figure}

\begin{figure}
  \includegraphics[width=\textwidth]{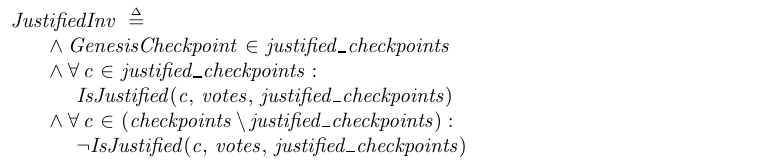}
  \caption{Justified checkpoint lemmas in \textit{IndInv}}\label{figJC}
\end{figure}

Since we represent a fork by means of a sign-change on block numbers, we have to specify that their absolute values are contiguous, that the chains coincide in the absence of a fork, as well as on the pre-fork prefix, and that both chains are monotone w.r.t. block numbers after the fork point.
Additionally, we require validity predicates for the vote set and checkpoints, as well as a precise characterization of the set of justified checkpoints. The latter merits further discussion.

\paragraph{Justified checkpoints.} To accurately describe the set of all justified checkpoints, we require two constraints:
\begin{enumerate}
	\item Consistency: Every checkpoint in the set is justified, and 
	\item Completeness: Every justified checkpoint belongs to the set.
\end{enumerate}
It is worth noting that both of these properties are required for an inductive invariant; if we do not specify consistency, the solver can trivially infer that the set contains all possible checkpoints, and that all checkpoints are finalized, which leads to bogus counterexamples to \textit{AccountableSafety}.
On the other hand, if we don't specify completeness, the solver can trivially infer that the set of justified checkpoints is empty (or contains exactly the genesis checkpoint). This leads us to be unable to detect real violations of \textit{AccountableSafety}, since we can never observe conflicting finalized checkpoints, even when the votes cast necessitate their existence.

This is critical, because including both constraints places a heavy burden on the solver. No matter how we define the justification predicate, it will inevitably appear in both positive and negative form, forcing the solver to contend with both quantifier alternations and double universal quantification, both of which are known to be hard.
Fundamentally, this demonstrates the intrinsic complexity of the problem itself, regardless of the particular characterization of justified sets in \tlap{}.

\subsection{Model Checking Experiments}

Table~\ref{tab:spec4-experiments} summarizes our experiments on checking
inductiveness and accountable safety of~\SpecFourB{}. Interestingly, checking
inductiveness of~\texttt{IndInv} takes seconds. However,
checking~\texttt{AccountableSafety} against the inductive invariant times out.

\begin{table}
    \centering
    \begin{tabular}{lllrrr}
            \tbh{Init}
            & \tbh{Step}
            & \tbh{Invariant}
            & \tbh{Depth}
            & \tbh{Memory}
            & \tbh{Time}
            \\ \toprule
            \texttt{Init}
            & \texttt{Next}
            & \texttt{IndInv}
            & 0
            & 0.6 GB
            & 2sec
            \\
            \texttt{IndInit}
            & \texttt{Next}
            & \texttt{IndInv}
            & 1
            & 0.6 GB
            & 2sec
            \\
            \texttt{IndInit}
            & \texttt{Next}
            & \texttt{AccountableSafety}
            & 0
            & 1.5 GB
            & TO ($> 6$d)
            \\
            \bottomrule
    \end{tabular}
    \caption{Model checking experiments
        with~\SpecFour{} on the instance~\texttt{MC\_ffg\_b3\_ffg5\_v12}
        (``TO'' means timeout)
    }\label{tab:spec4-experiments}
\end{table}

\section{\SpecFourB{}: Decomposition \& Abstractions}\label{sec:spec4b}

Disappointed with the results in Section~\ref{sec:spec4}, we decided to
push abstractions even further. These abstractions helped us to verify
accountable safety for the configuration in Figure~\ref{fig:three}.
Unfortunately, they do not scale to larger configurations. In any case, we
find these abstractions quite important for further research on model checking
of algorithms similar to 3SF\@.

\subsection{Constraints over Set Cardinalities}

\SpecThree{} contains several comparisons over set cardinalities. For example:

\begin{equation}
    3 * \textit{Cardinality}(\textit{validatorsWhoCastJustifyingVote}) \ge 2 \cdot N
    \label{eq:card-comparison}
\end{equation}

In the general case, the symbolic model checker has to encode constraints for
the cardinality computation in Equation~(\ref{eq:card-comparison}). If a set
$S$ contains up to $n$~elements, Apalache produces $O(n^2)$~constraints for
$\textit{Cardinality(S)}$.

To partially remediate the above issue in~\SpecFour{}, we introduce a
constant~$T$ for the upper bound on the number of faulty process, and further
refine Equation~(\ref{eq:card-comparison}) to:

\begin{equation}
    \textit{Cardinality}(\textit{validatorsWhoCastJustifyingVote}) \ge 2 \cdot T + 1
    \label{eq:card-comparison2}
\end{equation}

Apalache applies an optimized translation rule for
Equation~(\ref{eq:card-comparison2}). Essentially, the solver has to find $2
\cdot T + 1$ set elements to show that Equation~(\ref{eq:card-comparison2})
holds true. This gives us a linear number of constraints, instead of a
quadratic one. For example, when $T=1$ and the set may contain up to~$n$
elements, the model checker produces $O(n)$ constraints to check
Equation~(\ref{eq:card-comparison2}). A similar optimization is used in the
specification of Tendermint~\cite{TendermintSpec2020}.

\subsection{Quorum Sets}

To further optimize the constraints over set cardinalities, we have applied the
well-known pattern of replacing cardinality tests with quorum sets. For
example, this approach is used in the specification of Paxos~\cite{tla-paxos}.
To this end, we introduce quorum sets such as in Figure~\ref{fig:quorum-sets}.

\begin{figure}[!h]
    \includegraphics[width=\textwidth]{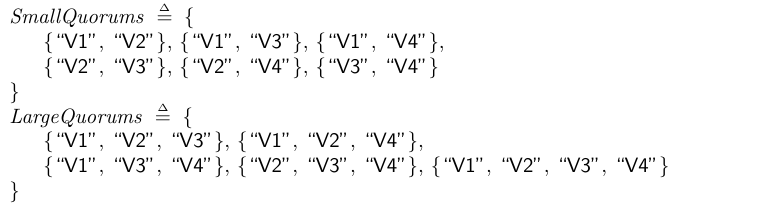}
    \caption{Quorum sets}\label{fig:quorum-sets}
\end{figure}

By using quorum sets, we further replace cardinality comparisons like in
Equation~(\ref{eq:card-comparison2}) with membership tests like in
Equation~(\ref{eq:card-comparison3}):

\begin{equation}
    \textit{validatorsWhoCastJustifyingVote} \in \textit{LargeQuorums}
    \label{eq:card-comparison3}
\end{equation}

\subsection{Decomposition of Chain Configurations}\label{sec:decomposition}

Recall Figure~\ref{figFork} from Section~\ref{sec:spec4-indinv}, which poses
constraints on the chains and the fork points. While these constraints should
be easier for an SMT solver than general reachability properties, they
still produce a number of arithmetic constraints. On the other hand, when we do
model checking for small parameters, there is a relatively small set of
possible chain configurations.  Figure~\ref{fig:block-graphs} shows some of
these configurations for the graphs of 3 to 7 blocks.

This observation led us to the following idea. Instead of using the constraints
over blocks such as in Figure~\ref{figFork}, we introduce one instance per
chain configuration. It is thus sufficient to verify accountable safety for all
of these instances and aggregate the model checking results.

Figure~\ref{fig:indinit-c3} shows an initialization predicate for the
configuration shown in Figure~\ref{fig:five2}. This predicate replaces the
general initialization predicate that we discussed in
Section~\ref{sec:spec4-indinv}.

\begin{figure}
    \includegraphics[width=\textwidth]{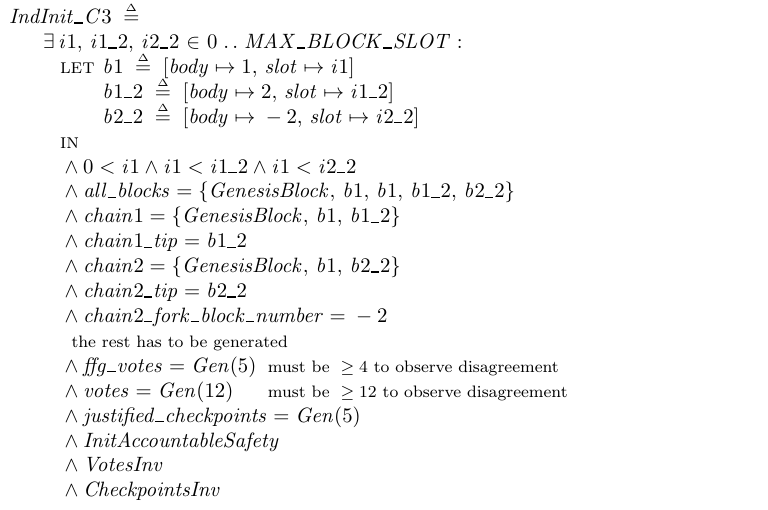}
    \caption{Initialization predicate for the
             configuration~\text{M5b}}\label{fig:indinit-c3}
\end{figure}

\subsection{Model Checking Experiments}

Tables~\ref{tab:spec4b-experiments} and~\ref{tab:spec4b-inductiveness}
summarize our experiments with Apalache for various configurations. One
interesting effect of the optimizations, especially of the ones presented in
Section~\ref{sec:decomposition}, is a significant drop in the memory
consumption of the SMT solver. In our experiments, Z3 required from 700~MB to
1.5~GB\@. While this is still a factor of 20 in comparison to the Alloy
experiments in Section~\ref{sec:alloy}, this is significantly better than our
initial experiments with~\SpecTwo{} and~\SpecThree{}, which required up to
20~GB of RAM\@. Interestingly, inductiveness checks in
Table~\ref{tab:spec4b-inductiveness} take significantly longer than in the case
of~\SpecFour{}. We conjecture that this is caused by the need to check more
specialized graphs in conjunction with steps.

We observe that the run-times depend on the different
configurations used. This is probably due to the effect of built-in variance of runs of the Z3
SMT solver, which is well-known in the
community of computer-aided verification. To further confirm these variations,
we could run multiple experiments with~\texttt{hyperfine}.

\begin{table}
    \centering
    \begin{tabular}{lllrr}
        \tbh{Configuration}
            & \tbh{Instance}
            & \tbh{Init}
            & \tbh{Memory}
            & \tbh{Time}
            \\ \toprule
        M3: Fig.~\ref{fig:three}
            & \texttt{MC\_ffg\_b1\_ffg5\_v12}
            & \texttt{Init\_C1}
            & 1.2 GB
            & 11h 31min
            \\
        M4a: Fig.~\ref{fig:four-top}
            & \texttt{MC\_ffg\_b3\_ffg5\_v12}
            & \texttt{Init\_C4}
            & 1.4 GB
            & 6d 15h
            \\
        M4b: Fig.~\ref{fig:four-bottom}
            & \texttt{MC\_ffg\_b3\_ffg5\_v12}
            & \texttt{Init\_C2}
            & 1.3 GB
            & 1d 6h
            \\
        M5a: Fig.~\ref{fig:five1}
            & \texttt{MC\_ffg\_b3\_ffg5\_v12}
            & \texttt{Init\_C3}
            & 1.2 GB
            & 1h 53min
            \\
        M5b: Fig.~\ref{fig:five2}
            & \texttt{MC\_ffg\_b3\_ffg5\_v12}
            & \texttt{Init\_C1}
            & 1.5 GB
            & 22d
            \\
            \bottomrule
    \end{tabular}
    \caption{Checking accountable safety against~\SpecFourB{}}\label{tab:spec4b-experiments}
\end{table}

\begin{table}
    \centering
    \begin{tabular}{lllrr}
        \tbh{Instance}
            & \tbh{Init}
            & \tbh{Invariant}
            & \tbh{Memory}
            & \tbh{Time}
            \\ \toprule
        \texttt{MC\_ffg\_b1\_ffg5\_v12}
            & \texttt{Init}
            & \texttt{IndInv}
            & 580 MB
            & 7s
            \\
        \texttt{MC\_ffg\_b3\_ffg5\_v12}
            & \texttt{Init}
            & \texttt{IndInv}
            & 700 MB
            & 7s
            \\
        \texttt{MC\_ffg\_b1\_ffg5\_v12}
            & \texttt{Init\_C1}
            & \texttt{IndInv}
            & 1.4 GB
            & 2min 8s
            \\
        \texttt{MC\_ffg\_b3\_ffg5\_v12}
            & \texttt{Init\_C1}
            & \texttt{IndInv}
            & 1.8 GB
            & 19min 10s
            \\
        \texttt{MC\_ffg\_b3\_ffg5\_v12}
            & \texttt{Init\_C2}
            & \texttt{IndInv}
            & 1.6 GB
            & 13min 16s
            \\
        \texttt{MC\_ffg\_b3\_ffg5\_v12}
            & \texttt{Init\_C3}
            & \texttt{IndInv}
            &  1.6 GB
            & 17min 39s
            \\
        \texttt{MC\_ffg\_b3\_ffg5\_v12}
            & \texttt{Init\_C4}
            & \texttt{IndInv}
            & 1.6 GB
            & 16min 23s
            \\
            \bottomrule
    \end{tabular}
    \caption{Checking inductiveness
             for~\SpecFourB{}}\label{tab:spec4b-inductiveness}
\end{table}

\section{Summary \& Future Work}\label{sec:summary}
We have presented a series of specifications modeling the 3SF protocol from
various perspectives. Initially, we developed a direct translation of the
protocol's Python specification into \tlap{}, but this approach proved
unsatisfactory due to the reliance on recursion. To address this, we modified
the specification to use folds in place of recursion, theoretically enabling
model-checking with Apalache. However, this approach also proved impractical
due to the high computational complexity involved. Subsequently, we applied a
series of abstractions to improve the model-checking efficiency.

In addition to the \tlap{} specifications, we also introduced an SMT encoding
and an Alloy specification. The SMT encoding proved to be fairly performant,
while the Alloy specification demonstrated exceptional performance in
combination with the Kissat SAT solver.

\subsection{Key Outcomes of the Project}

To revisit the key outcomes of the project, see Section~\ref{sec:outcomes}.

\subsection{Potential Extensions of this Project}\label{sec:future}

\begin{enumerate}

  \item \emph{Prove some of the other properties guaranteed by 3SF.} This project focused on verifying AccountableSafety, arguably the most critical property that the 3SF protocol must satisfy.
  However, it is also arguably the most challenging to verify through model checking, as it is the only property involving all honest participants that must hold true under fully asynchronous network conditions.
  In contrast, proving that honest nodes never commit slashable offenses (a non-distributed system property dependent only on the behavior of a single node) or properties reliant on network synchrony, such as reorg-resilience and dynamic-availability, is expected to be simpler, albeit this would require extending the \tlap{} encoding to include the behavior of honest nodes, which was unnecessary for verifying AccountableSafety.

  \item \emph{Generating inputs to the Python specification.} As we have noted,
    the power of our~\tlap{} specifications is the ability to generate examples
    with Apalache. This would help the authors of the Python specification to
    produce tests for their specifications.

  \item \emph{Specifications of a refined protocol.} The current version of the
    Python specification is very abstract. On one hand, it is usually
    beneficial to specify a high-level abstraction. On the other hand, as we
    found, the current level of abstraction is quite close to the general
    inductive definitions of justified and finalized checkpoints. We expect a
    refined protocol specification to be more amenable to model checking.

  \item \emph{Transferring the Alloy encoding to Apalache.} As we have found in
    this project, Alloy offers richer options for fine tuning in terms of the
    search scope. In combination with steady advances in SAT solving, adapting
    the Alloy encoding to Apalache would improve model checking performance,
    ultimately leading to a faster feedback loop and faster specification
    development.

\end{enumerate}

\subsection*{Acknowledgements}

We are grateful to Luca Zanolini and Francesco D'Amato for reviewing the earlier
versions of our report and discussing the 3SF Protocol with us.

\pagebreak
\bibliographystyle{plain}
\bibliography{ref}

\pagebreak

\appendix

The work done in this section is the main contribution of Milestone~3. Since
this section is quite long and technical, we have decided to add it in the
appendix.

\section{Translating Python Specifications to \tlap{}}\label{section3}

In this section, we present our results on translating a subset of Python that
is used to write executable specifications in the projects such as
\texttt{ssf-mc}\footnote{\url{https://github.com/saltiniroberto/ssf}}. Since we
have done the translation by hand, our rules are currently formalized on paper.
Additionally, in case of non-trivial rules, we give correctness proofs.

Since the Python subset uses the package $\texttt{pyrsistent}$, we assume that
the expressions are typed according to the package types, which can be found in
the $\texttt{typing}$ module. In the following, given a Pyrsistent type $\tau$,
we will denote its corresponding type in the Apalache type
system~1.2\footnote{\url{https://apalache-mc.org/docs/adr/002adr-types.html}}
with $\htau$. Table~\ref{tab:types} shows the types mapping.

\begin{table}[!h]
    \centering
    \begin{tabular}{cc}
        \tbh{Pyrsistent type} & \tbh{Apalache type}
            \\\toprule
        bool & Bool \\\midrule
        int & Int \\\midrule
        str & Str \\\midrule
        $\PSet[\tau]$ & $\Set(\htau)$ \\\midrule
        $\PVec[\tau]$ & $\List(\htau) \defeq \{ es\colon \Seq(\htau) \}$ \\\midrule
        $\PMap[\tau_1, \tau_2]$ & $\htau_1 \rightarrow \htau_2$ \\\midrule
        $\Callable[[\tau_1], \tau_2]$ & $\tau_1 \Rightarrow \tau_2$ \\\midrule
    \end{tabular}
    \caption{Mapping the Pyrsistent
             types to Apalache types}\label{tab:types}
\end{table}

Note that instead of using the standard type $\Seq(\htau)$ of \tlap{}, which
represents 1-indexed sequences, we use an alternative module
\texttt{Lists}\footnote{\url{https://github.com/konnov/tlaki/blob/main/src/Lists.tla}},
which represents 0-indexed sequences. To that end, we introduce the type
notation:

\[ \List(\htau) \coloneqq \{ es\colon \Seq(\htau) \} \]

The translation rules can be easily adapted to~$\Seq(\htau)$ instead
of~$\List(\htau)$.

\subsection{Translation Rules}

We give one translation rule per expression.

\subsubsection{Singleton vector}

\begin{mathpar}
    \inferrule*[right=(\textsc{Vec})]
    {
        \inferrule{a\colon \tau}{e \colon \htau}
        \\
        \mathrm{pvector\_of\_one\_element}(a)
    }{
        \List(\tup{e})\colon \List(\htau)
    }
\end{mathpar}

A singleton Python vector is translated to a single-element list, and annotated
as such. See
\href{https://github.com/saltiniroberto/ssf/blob/7ea6e18093d9da3154b4e396dd435549f687e6b9/high_level/common/pythonic_code_generic.py#L15-L16}{Source}.

\subsubsection{Vector concatenation}

\begin{mathpar}
    \inferrule*[right=(\textsc{Concat})]
    {
        \inferrule{a\colon \PVec[\tau]}{e \colon \List(\htau)}
        \\
        \inferrule{b\colon \PVec[\tau]}{f \colon \List(\htau)}
        \\
        \mathrm{pvector\_concat}(a, b)
    }{
        \Concat(e,f) \colon \List(\htau)
    }
\end{mathpar}

Vector concatenation is translated to the list concatenation. See
\href{https://github.com/saltiniroberto/ssf/blob/7ea6e18093d9da3154b4e396dd435549f687e6b9/high_level/common/pythonic_code_generic.py#L19-L20}{Source}.

\subsubsection{Set sequentialization}

\begin{mathpar}
    \inferrule*[right=(\textsc{SetToVec})]
    {
        \inferrule{a\colon \PSet[\tau]}{e \colon \Set(\htau)}
        \\
        s \coloneqq \tup{}\colon \Seq(\htau)
        \\
        \mathrm{from\_set\_to\_pvector}(a)
    }{
        \ApaFoldSet( \Push, \List(s), e ) \colon \List(\htau)
    }
\end{mathpar}

We use fold, to create a sequence (in some order) from the set.  See
\href{https://github.com/saltiniroberto/ssf/blob/7ea6e18093d9da3154b4e396dd435549f687e6b9/high_level/common/pythonic_code_generic.py#L23-L24}{Source}.

\subsubsection{ Empty set}

\begin{mathpar}
    \inferrule*[right=(\textsc{EmptySet})]
    {
        \mathrm{pset\_get\_empty}() \colon \PSet[t]
    }{
        \{\} \colon \Set(\htau)
    }
\end{mathpar}

The only relevant part here is that we need a type annotation on the Python
side to correctly annotate the empty set in \tlap{}.  See
\href{https://github.com/saltiniroberto/ssf/blob/7ea6e18093d9da3154b4e396dd435549f687e6b9/high_level/common/pythonic_code_generic.py#L27-L28}{Source}.

\subsubsection{ Set union}

\begin{mathpar}
    \inferrule*[right=(\textsc{Union})]
    {
        \inferrule{a\colon \PSet[\tau]}{e \colon \Set(\htau)}
        \\
        \inferrule{b\colon \PSet[\tau]}{f \colon \Set(\htau)}
        \\
        \mathrm{pset\_merge}(a, b)
    }{
        e \cup f \colon \Set(\htau)
    }
\end{mathpar}

Set union is translated to the \tlap{} native set union.  See
\href{https://github.com/saltiniroberto/ssf/blob/7ea6e18093d9da3154b4e396dd435549f687e6b9/high_level/common/pythonic_code_generic.py#L31-L32}{Source}.

\subsubsection{ Set flattening}

\begin{mathpar}
    \inferrule*[right=(\textsc{BigUnion})]
    {
        \inferrule{ a\colon \PSet[\PSet[\tau]]}{e \colon \Set(\Set(\htau))}
        \\
        \mathrm{pset\_merge\_flatten}(a)
    }{
        \UNION e \colon \Set(\htau)
    }
\end{mathpar}

Set flattening is translated to the \tlap{} native big $\UNION$.  See
\href{https://github.com/saltiniroberto/ssf/blob/7ea6e18093d9da3154b4e396dd435549f687e6b9/high_level/common/pythonic_code_generic.py#L35-L36}{Source}.

\subsubsection{Set intersection}

\begin{mathpar}
    \inferrule*[right=(\textsc{Intersection})]
    {
        \inferrule{a\colon \PSet[\tau]}{e \colon \Set(\htau)}
        \\
        \inferrule{b\colon \PSet[\tau]}{f \colon \Set(\htau)}
        \\
        \mathrm{pset\_intersection}(a, b)
    }{
        e \cap f \colon \Set(\htau)
    }
\end{mathpar}

Set intersection is translated to the \tlap{} native set intersection.  See
\href{https://github.com/saltiniroberto/ssf/blob/7ea6e18093d9da3154b4e396dd435549f687e6b9/high_level/common/pythonic_code_generic.py#L42-L43}{Source}.

\subsubsection{Set difference}

\begin{mathpar}
    \inferrule*[right=(\textsc{SetDiff})]
    {
        \inferrule{a\colon \PSet[\tau]}{e \colon \Set(\htau)}
        \\
        \inferrule{b\colon \PSet[\tau]}{f \colon \Set(\htau)}
        \\
        \mathrm{pset\_difference}(a, b)
    }{
        e \setminus f \colon \Set(\htau)
    }
\end{mathpar}

Set difference is translated to the \tlap{} native set difference.
\href{https://github.com/saltiniroberto/ssf/blob/7ea6e18093d9da3154b4e396dd435549f687e6b9/high_level/common/pythonic_code_generic.py#L46-L47}{Source}.

\subsubsection{Singleton set}

\begin{mathpar}
    \inferrule*[right=(\textsc{Singleton})]
    {
        \inferrule{a\colon \tau}{e \colon \htau}
        \\
        \mathrm{pset\_get\_singleton}(a)
    }{
        \{e\} \colon \Set(\htau)
    }
\end{mathpar}

A singleton Python set is translated to a \tlap{} native single-element set.
See
\href{https://github.com/saltiniroberto/ssf/blob/7ea6e18093d9da3154b4e396dd435549f687e6b9/high_level/common/pythonic_code_generic.py#L50-L51}{Source}.

\subsubsection{Set extension}

\begin{mathpar}
    \inferrule*[right=(\textsc{SetExt})]
    {
        \inferrule{a\colon \PSet[\tau]}{e \colon \Set(\htau)}
        \\
        \inferrule{b\colon \tau}{f \colon \htau}
        \\
        \mathrm{pset\_add}(a, b)
    }{
        e \cup \{f\} \colon \Set(\htau)
    }
\end{mathpar}
A set extension is translated to a combination of union and singleton-set construction. Semantically, this is the equivalence:

\[
\mathrm{pset\_add}(a,b) = \mathrm{pset\_merge}(a, \mathrm{pset\_get\_singleton}(b))
\]

See
\href{https://github.com/saltiniroberto/ssf/blob/7ea6e18093d9da3154b4e396dd435549f687e6b9/high_level/common/pythonic_code_generic.py#L54-L55}{Source}.

\subsubsection{Element choice}

\begin{mathpar}
    \inferrule*[right=(\textsc{Choice})]
    {
        \inferrule{a\colon \PSet[\tau]}{e \colon \Set(\htau)}
        \\
        \mathrm{pset\_pick\_element}(a)
    }{
        (\CHOOSE x \in e\colon\TRUE) \colon\htau
    }
\end{mathpar}

We translate this to the built in deterministic choice in \tlap{}. We cannot
account for the dynamic non-emptiness requirement. Instead, in that scenario,
the value of this expression is some unspecified element of the correct type.
See
\href{https://github.com/saltiniroberto/ssf/blob/7ea6e18093d9da3154b4e396dd435549f687e6b9/high_level/common/pythonic_code_generic.py#L58-L60}{Source.}

\subsubsection{Set filter}

\begin{mathpar}
    \inferrule*[right=(\textsc{Filter})]
    {
        \inferrule{a\colon \Callable[[\tau], \bool]}{e \colon \htau \to \Bool}
        \\
        \inferrule{b\colon \PSet[\tau]}{f \colon \Set(\htau)}
        \\
        \mathrm{pset\_filter}(a, b)
    }{
        \{ x \in f \colon e[x] \} \colon \Set(\htau)
    }
\end{mathpar}

Set filtering is translated to the \tlap{} native filter operation.  See
\href{https://github.com/saltiniroberto/ssf/blob/7ea6e18093d9da3154b4e396dd435549f687e6b9/high_level/common/pythonic_code_generic.py#L63-L70}{Source}.

\subsubsection{ Set maximum}

\begin{mathpar}
    \inferrule*[right=(\textsc{Max})]
    {
        \inferrule{a\colon \PSet[\tau]}{e \colon \Set(\htau)}
        \\
        \inferrule{b\colon \Callable[[\tau], T]}{f \colon \htau \to \hat{T}}       
        \\
        \mathrm{pset\_max}(a, b)
    }{
        (\CHOOSE m \in e\colon \forall x \in e\colon \Le(f[x], f[m]))\colon \htau
    }
\end{mathpar}

See
\href{https://github.com/saltiniroberto/ssf/blob/7ea6e18093d9da3154b4e396dd435549f687e6b9/high_level/common/pythonic_code_generic.py#L74-L76}{Source}.

Here, the translation depends on the type $T$ (resp. type $\hat{T}$), since there is no built-in notion of ordering in \tlap{}. 
\paragraph{Instance 1:} If $\hat{T}$ is an integer type, then 
\begin{lstlisting}[language=tla,columns=fullflexible]
Le(x,y) $\defeq$ x $\le$ y
\end{lstlisting}
\paragraph{Instance 2:} If $\hat{T}$ is a tuple type $\tup{\Int,\Int}$, it is instead 
\begin{lstlisting}[language=tla,columns=fullflexible]
Le(x,y) $\defeq$ 
  IF x[1] > y[1]
  THEN FALSE
  ELSE IF x[1] < y[1]
       THEN TRUE
       ELSE x[2] $\le$ y[2]
\end{lstlisting}

\subsubsection{ Set sum}

\begin{mathpar}
    \inferrule*[right=(\textsc{Sum})]
    {
        \inferrule{a\colon \PSet[\pyint]}{e \colon \Set(\Int)}
        \\
        \mathrm{pset\_sum}(a)
    }{
        \ApaFoldSet(+, 0, e) \colon \Int
    }
\end{mathpar}

We translate a set sum as a fold of the $+$ operator over the set. See
\href{https://github.com/saltiniroberto/ssf/blob/7ea6e18093d9da3154b4e396dd435549f687e6b9/high_level/common/pythonic_code_generic.py#L79-L80}{Source}.

\subsubsection{ Set emptiness check}

\begin{mathpar}
    \inferrule*[right=(\textsc{IsEmpty})]
    {
        \inferrule{a\colon \PSet[\tau]}{e \colon \Set(\htau)}
        \\
        \mathrm{pset\_is\_empty}(a)
    }{
        e = \{\} \colon \Bool
    }
\end{mathpar}

The emptiness check is translated to a comparison with the explicitly
constructed empty set. See
\href{https://github.com/saltiniroberto/ssf/blob/7ea6e18093d9da3154b4e396dd435549f687e6b9/high_level/common/pythonic_code_generic.py#L83-L84}{Source}.

\subsubsection{Vector-to-Set conversion}

\begin{mathpar}
    \inferrule*[right=(\textsc{VecToSet})]
    {
        \inferrule{a\colon \PVec[\tau]}{e \colon \List(\htau)}
        \\
        \mathrm{from\_pvector\_to\_pset}(a)
    }{
        \{ \At(e, i)\colon i \in \Indices(e) \} \colon \Set(\htau)  
    }
\end{mathpar}

We translate the set-conversion, by mapping the accessor method over
$\Indices$. See
\href{https://github.com/saltiniroberto/ssf/blob/7ea6e18093d9da3154b4e396dd435549f687e6b9/high_level/common/pythonic_code_generic.py#L87-L88}{Source}.

\subsubsection{Set mapping}

\begin{mathpar}
    \inferrule*[right=(\textsc{Map})]
    {
        \inferrule{a\colon \Callable[[\tau_1], \tau_2]}{e \colon \htau_1 \to \htau_2}
        \\
        \inferrule{b\colon \PSet[\tau_1]}{f \colon \Set(\htau_1)}
        \\
        \mathrm{pset\_map}(a, b)
    }{
        \{ e[x]\colon x \in f\} \colon \Set(\htau_2)
    }
\end{mathpar}

Set mapping is translated to the \tlap{} native map operation. See
\href{https://github.com/saltiniroberto/ssf/blob/7ea6e18093d9da3154b4e396dd435549f687e6b9/high_level/common/pythonic_code_generic.py#L91-L97}{Source}.

\subsubsection{Function domain inclusion check}

\begin{mathpar}
    \inferrule*[right=(\textsc{InDom})]
    {
        \inferrule{a\colon \PMap[\tau_1, \tau_2]}{f \colon \htau_1 \to \htau_2}
        \\
        \inferrule{b\colon \tau_1}{e \colon \htau_1}
        \\
        \mathrm{pmap\_has}(a, b)
    }{
        e \in \DOMAIN f\colon \Bool
    }
\end{mathpar}

Function domain inclusion checking is translated to the \tlap{} native
set-inclusion operation for $\DOMAIN$.  See
\href{https://github.com/saltiniroberto/ssf/blob/7ea6e18093d9da3154b4e396dd435549f687e6b9/high_level/common/pythonic_code_generic.py#L100-L101}{Source}.

\subsubsection{ Function application}

\begin{mathpar}
    \inferrule*[right=(\textsc{App})]
    {
        \inferrule{a\colon \PMap[\tau_1, \tau_2]}{f \colon \htau_1 \to \htau_2}
        \\
        \inferrule{b\colon \tau_1}{e \colon \htau_1}
        \\
        \mathrm{pmap\_get}(a, b)
    }{
        f[e]\colon \htau_2
    }
\end{mathpar}

Function application is translated to the \tlap{} native function application.
We cannot account for the dynamic domain-membership requirement. Instead, in
that scenario, the value of this expression is some unspecified element of the
correct type. See
\href{https://github.com/saltiniroberto/ssf/blob/7ea6e18093d9da3154b4e396dd435549f687e6b9/high_level/common/pythonic_code_generic.py#L104-L106}{Source}.

\subsubsection{ Empty function}

\begin{mathpar}
    \inferrule*[right=(\textsc{EmptyFun})]
    {
        \mathrm{pmap\_get\_empty}()\colon \PMap[\tau_1,\tau_2]
        \\
        s \coloneqq \{\}\colon \Set(\tup{\htau_1, \htau_2}) 
    }{
      	\SetAsFun(s)\colon \htau_1 \to \htau_2
    }
\end{mathpar}

See
\href{https://github.com/saltiniroberto/ssf/blob/7ea6e18093d9da3154b4e396dd435549f687e6b9/high_level/common/pythonic_code_generic.py#L109-L110}{Source}.

We use Apalache's $\SetAsFun$, since we only need to annotate the empty set
with the correct tuple type. The native construction via $[ \_ \mapsto \_]$
would require us to invent a codomain value, which we might not have access to
if $\tau_1 \ne \tau_2$.

\subsubsection{ Function update}

\begin{mathpar}
    \inferrule*[right=(\textsc{Update})]
    {
        \inferrule{a\colon \PMap[\tau_1, \tau_2]}{f \colon \htau_1 \to \htau_2}
        \\
        \inferrule{b\colon \tau_1}{x \colon \htau_1}
        \\
        \inferrule{c\colon \tau_2}{y \colon \htau_2}
        \\
        \mathrm{pmap\_set}(a, b, c)
    }{
        [ v \in (\DOMAIN f \cup \{x\}) \mapsto \IF v = x \THEN y \ELSE f[x] ] \colon \htau_1 \to \htau_2
    }
\end{mathpar}

See
\href{https://github.com/saltiniroberto/ssf/blob/7ea6e18093d9da3154b4e396dd435549f687e6b9/high_level/common/pythonic_code_generic.py#L113-L114}{Source}.

While one might intuitively want to translate map updates using the \tlap{}
native $\EXCEPT$, we cannot, since $\EXCEPT$
\href{https://lamport.azurewebsites.net/tla/book-21-07-04.pdf}{does not allow
for domain extensions}, whereas $\mathrm{pmap\_set}$
\href{https://pyrsistent.readthedocs.io/en/latest/api.html#pyrsistent.PMap.set}{does}.
By definition:

\[
[f \EXCEPT\ ![x] = y] \defeq [v \in \DOMAIN f \mapsto \IF v = x \THEN y \ELSE f[x] ]
\]

where, most notably, the domain of $[f \EXCEPT ![x] = y]$ is exactly the domain of $f$. We adapt the above definition to (possibly) extend the domain.

\subsubsection{Function combination}

\begin{mathpar}
    \inferrule*[right=(\textsc{FnMerge})]
    {
        \inferrule{a\colon \PMap[\tau_1, \tau_2]}{f \colon \htau_1 \to \htau_2}
        \\
        \inferrule{b\colon \PMap[\tau_1, \tau_2]}{g \colon \htau_1 \to \htau_2}
        \\
        \mathrm{pmap\_merge}(a,b)
    }{
        [x \in (\DOMAIN f \cup \DOMAIN g) \mapsto \IF x \in \DOMAIN g \THEN g[x] \ELSE f[x]]\colon \htau_1\to\htau_2
    }
\end{mathpar}

See~\href{https://github.com/saltiniroberto/ssf/blob/7ea6e18093d9da3154b4e396dd435549f687e6b9/high_level/common/pythonic_code_generic.py#L117-L118}{Source}.
Function combination is translated to a new function, defined over the union of
both domains. Note that the second map/function dominates in the case of
key/domain collisions.

\subsubsection{Function domain}

\begin{mathpar}
    \inferrule*[right=(\textsc{FnDomain})]
    {
        \inferrule{a\colon \PMap[\tau_1, \tau_2]}{e \colon \htau_1 \to \htau_2}
        \\
        \mathrm{pmap\_keys}(a)
    }{
        \DOMAIN e\colon \Set(\htau_1)
    }
\end{mathpar}

See
\href{https://github.com/saltiniroberto/ssf/blob/7ea6e18093d9da3154b4e396dd435549f687e6b9/high_level/common/pythonic_code_generic.py#L121-L122}{Source}.
We translate this to the \tlap{} native $\DOMAIN$.

\subsubsection{ Function codomain}

\begin{mathpar}
    \inferrule*[right=(\textsc{FnCodomain})]
    {
        \inferrule{a\colon \PMap[\tau_1, \tau_2]}{e \colon \htau_1 \to \htau_2}
        \\
        \mathrm{pmap\_values}(a)
    }{
        \{ e[x]\colon x \in \DOMAIN e \}\colon \Set(\htau_2)
    }
\end{mathpar}

See~\href{https://github.com/saltiniroberto/ssf/blob/7ea6e18093d9da3154b4e396dd435549f687e6b9/high_level/common/pythonic_code_generic.py#L125-L126}{Source}.
We translate this by mapping the function over its $\DOMAIN$.
 
\subsection{Meta-Rules for Translating Recursive Definitions}
\label{subsec:recrules}

So far, we have been dealing with the primitive building blocks. In order to
facilitate translation of recursive Python code to the \tlap{} fragment
supported by Apalache, we introduce a set of \tlap{}-to-\tlap{} rules, which
allow us to:

\begin{enumerate}
  \item Formulate translations from Python to \tlap{} in the
    intuitive way, potentially introducing constructs like recursion, and then

  \item Pair them with a \tlap{}-to-\tlap{} rule, ending in a supported
      fragment.

\end{enumerate}

We have to do that, as the \texttt{ssf} code extensively uses recursive
definitions.

\subsubsection{Bounded recursion rule}

Assume we are given a $\RECURSIVE$ operator $\op$. Without loss of generality,
we can take the arity to be $1$, since any operator of higher arity can be
expressed as an arity $1$ operator over tuples or records.

Further, we assume that the operator $\op$ has the following shape:

\begin{lstlisting}[language=tla,columns=fullflexible]
RECURSIVE R(_)
\* $@$type (a) => b;
R(x) ==
  IF P(x)
  THEN e
  ELSE G(x, R(next(x))
\end{lstlisting}

In other words, we have:

\begin{itemize}
  \item A termination condition $P$
  \item A "default" value $e$, returned if the argument satisfies the termination condition
  \item A general case operator $G$, which invokes a recursive computation of $R$ over a modified parameter given by the operator $\bb$.
\end{itemize}

The following needs to hold true, to ensure recursion termination: For every
$x\colon a$, there exists a finite sequence $x = v_1, \dots, v_n$
that satisfies the following conditions:

\begin{itemize}
\item $P(v_n)$ holds
\item $v_{i+1} = \bb(v_i)$ for all $1 \le i < n$
\item $P(v_i)$ does not hold for any $1 \le i < n$ (i.o.w., this is the shortest sequence with the above two properties)
\end{itemize}

We will attempt to express the recursive operator $\op$ with a non-recursive
``iterative'' operator $\nrop$ of arity $2$, which takes an additional
parameter: a constant $N$. The non-recursive operator will have the property
that, for any particular choice of $x$, $\nrop(x, N)$ will evaluate to $\op(x)$
if $n < N$ (i.e., if the recursion stack of $\op$ has height of at most $N$).

To that end, we first define:
\begin{lstlisting}[language=tla,columns=fullflexible]
\* $@$type (a, Int) => Seq(a);
Stack(x, N) ==
  LET 
    \* $@$type: (Seq(a), Int) => Seq(a);
    step(seq, i) ==
      IF i > Len(seq) \/ P(seq[1])
      THEN seq
      ELSE 
        \* Alternatively, we can append here and reverse the list at the end
        <<next(seq[1])>> \o seq 
  IN ApaFoldSeqLeft( step, <<x>>, MkSeq(N, LAMBDA i: i) )
\end{lstlisting}

Here, $\ApaFoldSeqLeft$ is the left-fold over sequences, that is, an operator
for which:

\begin{align*}
\ApaFoldSeqLeft(O, v, \tup{}) =& v \\
\ApaFoldSeqLeft(O, v, \tup{x_1,\dots, x_n}) =& O(O(O(v, x_1), x_2), \dots, x_n)
\end{align*}
and $\MkSeq$ is the sequence constructor, for which 
\[
\MkSeq(N, O) = \tup{O(1), \dots, O(N)}
\]

We can see that $\Chain(x,N)$ returns the sequence $\tup{v_n, ..., x}$ if $N$
is sufficiently large.  We can verify whether or not that is the case, by
evaluating $P(\Chain(x, N)[1])$. If it does not hold, the $N$ chosen is not
large enough, and needs to be increased. Using $\Chain$ we can define a
fold-based non-recursive operator $\op^*$, such that $\op^*(x) = \op(x)$ under
the above assumptions:

\begin{lstlisting}[language=tla,columns=fullflexible]
\* $@$type (a, Int) => b;
I(x, N) ==
  LET stack == Stack(x, N) IN
  LET step(cumul, v) == G(v, cumul) IN
  ApaFoldSeqLeft( step, e, Tail(stack) )
\end{lstlisting}
Then, $\op^*(x) = \nrop(x, N_0)$ for some sufficiently large specification-level constant $N_0$. Alternatively,
\begin{lstlisting}[language=tla,columns=fullflexible]
\* $@$type (a, Int) => b;
I(x, N) ==
  LET stack == Stack(x, N) IN
  LET step(cumul, v) == G(v, cumul) IN
  IF P(stack[1])
  THEN ApaFoldSeqLeft(step, e, Tail(stack))
  ELSE CHOOSE x \in {}: TRUE 
\end{lstlisting}
In this form, we return $\CHOOSE x \in {}: \TRUE$, which is an idiom meaning "any value" (of the correct type), in the case where the $N$ chosen was not large enough. Tools can use this idiom to detect that $\nrop(x,N)$ did not evaluate to the expected value of $\op(x)$. 

\paragraph{Example.} Consider the following operator:
\begin{lstlisting}[language=tla,columns=fullflexible]
RECURSIVE R(_)
\* $@$type (Int) => Int;
R(x) ==
  IF x <= 0
  THEN 0
  ELSE x + R(x-1)
\end{lstlisting}
where $P(x) = x \le 0$, $G(a,b) = a + b$, and $\bb(x) = x - 1$. For this operator, we know that $\op(4) = 10$. By the above definitions:
\begin{lstlisting}[language=tla,columns=fullflexible]
\* $@$type (Int, Int) => Seq(Int);
Stack(x, N) ==
  LET 
    \* $@$type: (Seq(Int), Int) => Seq(Int);
    step(seq, i) ==
      IF i > Len(seq) \/ seq[1] <= 0
      THEN seq
      ELSE <<seq[1] - 1>> \o seq
  IN ApaFoldSeqLeft( step, <<x>>, MkSeq(N, LAMBDA i: i) )
\end{lstlisting}
We can compute the above $\Chain$ with two different constants $N$, $2$ and $100$, and observe that $\Chain(4, 2) = \tup{2, 3, 4}$ and $\Chain(4, 100) = \tup{0, 1, 2, 3, 4}$. 
We are able to tell whether we have chosen sufficiently large values for $N$ after the fact, by evaluating $P(\Chain(x,N)[1])$. 
For our $P(x) = x \le 0$, we see $\neg P(\Chain(4, 2)[1])$, and $P(\Chain(4, 100)[1])$, so we can conclude that we should not pick $N=2$, but $N=100$ suffices. 
Of course it is relatively easy to see, in this toy example, that the recursion depth is exactly $4$, but we could use this post-evaluation in cases where the recursion depth is harder to evaluate from the specification, to determine whether we need to increase the value of $N$. \hfill $\triangleleft$

\noindent Continuing with the next operator:
\begin{lstlisting}[language=tla,columns=fullflexible]
\* $@$type (Int, Int) => Int;
I(x, N) ==
  LET stack == Stack(x, N)
  IN ApaFoldSeqLeft( +, 0, Tail(stack) )
\end{lstlisting}
we see that $\nrop(4, 2) = 7 \ne 10 = \op(4)$, but $\nrop(4, 100) = 10 = \op(4)$.
As expected, choosing an insufficiently large value of $N$ will give us an incorrect result, but we know how to detect whether we have chosen an appropriate $N$.

\paragraph{Optimization for associative $G$.}
In the special case where $G$ is associative, that is, $G(a, G(b, c)) = G(G(a, b), c)$ for all $a,b,c$, we can make the entire translation more optimized, and single-pass. Since $\nrop(x,N)$, for sufficiently large $N$, computes 
\[
G(v_1, G(v_2, ... (G(v_{n-2}, G(v_{n-1}, e)))))
\]
and $G$ is associative by assumption, then computing
\[
G(G(G(G(v_1, v_2), ...), v_{n-1}), e)
\]
gives us the same value. This computation can be done in a single pass:
\begin{lstlisting}[language=tla,columns=fullflexible]
IForAssociative(x, N) ==
  IF P(x)
  THEN e
  ELSE
    LET 
      \* $@$type: (<<a, a>>, Int) => <<a, a>>;
      step(pair, i) == \* we don't use the index `i`
        LET partial == pair[1]
             currentElem == pair[2]
        IN
          IF P(currentElem)
          THEN pair
          ELSE
            LET nextElem == next(currentElem)
            IN << G(partial, IF P(nextElem) e ELSE nextElem), nextElem >>
    IN ApaFoldSeqLeft( step, <<x, x>>, MkSeq(N, LAMBDA i: i) )[1]
\end{lstlisting}

\subsubsection{ Mutual recursion cycles}

Assume we are given a collection of $n$ operators $\op_1, \dots, \op_n$ (using the convention $\op_{n+1} = \op_1$), with types $\op_i\colon (a_i) \Rightarrow a_{i+1}$ s.t. $a_{n+1} = a_1$, in the following pattern:

\begin{lstlisting}[language=tla,columns=fullflexible]
RECURSIVE R_i(_)
\* $@$type (a_i) => a_{i+1};
R_i(x) == G_i(x, R_{i+1}(next_i(x)))
\end{lstlisting}
Then, we can inline any one of these operators, w.l.o.g. $\op_1$, s.t. we obtain a primitive-recursive operator:
\begin{lstlisting}[language=tla,columns=fullflexible]
RECURSIVE R(_) 
\* $@$type: (a_1) => a_1;
R(x) ==
  G_1( x, 
    G_2( next_1(x),
      G_3( next_2(next_1(x)),
        G_4( ...
          G_n( next_{n-1}(next_{n-2}(...(next_1(x)))), 
               R(next_n(next_{n-1}(...(next_1(x)))))
            )
          )
        )
      )
    )
\end{lstlisting}
for which $\op(x) = \op_1(x)$ for all $x$, and $\op$ terminates iff $\op_1$ terminates.

\subsubsection{One-to-many recursion}

Suppose we are given, for each value $x: a$, a finite set $V(x)\colon \Set(a)$ s.t. $V(x)$ is exactly the set of values $v$, for which we are required to recursively compute $\op(v)$, in order to compute $\op(x)$. 
Further, assume that there exists a potential function $\gamma$ from $a$ to nonnegative integers, with the property that, for any $x$ of type $a$ the following holds:
\[
\forall y \in V(x)\colon \gamma(y) < \gamma(x) 
\]
If one cannot think of a more intuitive candidate for $\gamma$, one may always take $\gamma(t)$ to be the recursion stack-depth required to compute $\op(t)$ (assuming termination). It is easy to see that such a definition satisfies the above condition.

\paragraph{Example.} In the 3SF example, for each checkpoint $x$, the set $V(x)$ would be the set of all source-checkpoints belonging to FFG votes, which could be used to justify $x$ (and those checkpoints need to be recursively justified, forming a chain all the way back to genesis). Additionally, $\gamma$ would assign each checkpoint $x$ the value $x.\text{slot}$. \hfill $\triangleleft$

\noindent Let $\op$ have the following shape:
\begin{lstlisting}[language=tla,columns=fullflexible]
RECURSIVE R(_)
\* $@$type (a) => b;
R(x) ==
  IF P(x)
  THEN e
  ELSE G(x, F(V(x), Op))
\end{lstlisting}
where $F(S, T(\_)) \coloneqq \{s \in S\colon T(s)\}$ or $F(S, T(\_)) \coloneqq \{T(s)\colon s \in S\}$ (i.e. a map or a filter).

\paragraph{Example.} In the 3SF example, for $\mathrm{is\_justified\_checkpoint}$, $P$ checks whether $x$ is the genesis checkpoint, and $e = \TRUE$. $G$ is the main computation, which determines whether or not a non-genesis checkpoint is justified, by finding quorums of validators, wherein each validator cast an FFG vote justifying $x$, but where the source checkpoint was recursively justified (i.e. belonged to $V(x)$). \hfill $\triangleleft$
\\
\\
We give a translation scheme, which reduces this more generalized form of recursion to the one given in the previous section, by defining a map-based recursive operator $\mop$, for which we will ensure
\[
\op(x) = \mop([ v \in \{x\} \mapsto V(x) ])[x]
\] 
if $\op(x)$ terminates. We define the necessary operators in Figure \ref{fig1}. Using these operator definitions, we can show the following theorem:
\begin{figure}[ht]
\begin{lstlisting}[language=tla,columns=fullflexible]
\* $@$type: (a -> Set(a)) => a -> Set(a);
next(map) ==
  LET newDomain == UNION {map[v]: v \in DOMAIN map}
  IN [ newDomainElem \in newDomain |-> V(newDomainElem) ]

\* $@$type: (a -> Set(a), a -> b) => a -> b;
Gm(currentRecursionStepMap, partialValueMap) ==
  LET domainExtension == DOMAIN currentRecursionStepMap IN
  LET 
    \* $@$type: (a) => b;
    evalOneKey(k) ==
      LET OpSubstitute(x) == partialValueMap[x] 
      IN G(k, F(currentRecursionStepMap[k], OpSubstitute))
  IN [
    x \in (domainExtension \union DOMAIN partialValueMap) |->
      IF x \in DOMAIN partialValueMap
      THEN partialValueMap[x]
      ELSE IF P(x)
           THEN e
           ELSE evalOneKey(x)
  ]

RECURSIVE Rm(_)
\* $@$type (a -> Set(a)) => a -> b;
Rm(map) ==
  IF \A x \in DOMAIN map: P(x)
  THEN [ x \in DOMAIN map |-> e ]
  ELSE Gm(map, Rm(next(map)))
\end{lstlisting}
\caption{$\mop$ and auxiliary operators \label{fig1}}
\end{figure}

\newcommand{\thmBody}{
Let $f$ be a function, s.t. for any $x \in \DOMAIN f$ it is the case that $f[x] = V(x)$. Then, for
$g \coloneqq \mop(f)$:
\[
\forall x \in \DOMAIN g \colon g[x] = \op(x)
\]
}
 
\begin{theorem}\label{thm}
\thmBody
\end{theorem}
From this theorem, the following corollary trivially follows:
\newcommand{\corollaryBody}{
For any $x\colon a$
\[
\op(x) = \mop([ v \in \{x\} \mapsto V(x) ])[x]
\]
}
 
\begin{corollary}\label{corollary}
\corollaryBody
\end{corollary}
The equivalence proofs are available in the appendix. 
The termination proof must still be made on a case-by-case basis, as it depends on $h$ and $V$.

\section{Detailed Proofs}\label{proofs}

In the following, we show soundness of our translation rules for the recursive
operators.

\subsection{Additional Definitions}
Let $f$ be any \tlap function. We use the shorthand $D_f \coloneqq \DOMAIN f$.
We use $\nat$ to refer to the set of all natural numbers (i.e. nonnegative integers).

We assume the following precondition: There exists a $\gamma\colon a \to \nat$, with the following property:
\[
\forall x\colon a \ .\ \forall y \in V(x) \ .\ \gamma(y) < \gamma(x) 
\]

Recall the following definitions:
\[
\iteDef{\op(x)}{e}{P(x)}{G(x, F(V(x), \op))},
\]
\[
\iteDef{\mop(f)}{[ x \in D_f \mapsto e ]}{\forall x \in D_f \ .\ P(x)}{\mapg(f, \mop(\bb(f))}
\]
and
\[
\mapg(f, m)[x] \coloneqq \left\{
\begin{array}{ll}
      m[x] &; x \in D_m \\
      e &; x \notin D_m \land P(x) \\
      G(x, F(f[x], m)) &; \text{otherwise}\\
\end{array} 
\right. 
\]
with $D_{\mapg(f,m)} = D_f \cup D_m$.

\subsection{Proofs}

\begin{lemma}\label{lemma1}
Let $f$ be any function of type $a \to Set(a)$. Then
\[
D_f \subseteq D_{\mop(f)}
\]
\end{lemma}
\begin{proof}

Let $g \coloneqq \mop(f)$. To prove the first part, we have two options to
consider, based on the definition of $\op$. If $\forall x \in D_f \ .\ P(x)$,
then $D_f = D_g$. Otherwise, $g = \mapg(f,\mop(\bb(f)))$. By definition,
$D_{\mapg(f,m)} = D_f \cup D_m$, so it follows that $D_f \subseteq
D_{\mapg(f,m)}$, for any $m$. 
\end{proof}

\begin{corollary}
If $D_f = \emptyset$, then $D_{g} = \emptyset$.
\end{corollary}
\begin{proof}
If $D_f$ is empty, then it is vacuously true that $\forall x \in D_f \ .\ P(x)$, and so $D_f = D_g$ by definition.
\end{proof}

We define an auxiliary function $\alpha$ that assigns every function of the type $a \to Set(a)$ a value in $\nat \cup \{-\infty\}$, defined as:
\[
\alpha(f) \coloneqq \sup\left\{ \gamma(v) \mid v \in D_f \right\}
\]

\begin{lemma}\label{lemma2}
Let $f$ be a function, s.t. for any $x \in D_f$ it is the case that $f[x] = V(x)$. Then
\[
\alpha(f) \ge 0 \Rightarrow \alpha(\bb(f)) < \alpha(f)
\]
\end{lemma}

\begin{proof}

Assume $\alpha(f) \ge 0$. This is trivially equivalent to saying $D_f \ne
\emptyset$.  We see that $\bb$ is defined as:

\[
\bb(f)[y] \coloneqq V(y)
\]

$D_{\bb(f)}$ is defined as:

\[
D_{\bb(f)} \coloneqq \bigcup_{v \in D_f} f[v]
\]

It follows that:

\begin{align*} 
\alpha(\bb(f)) &= \sup\left\{ \gamma(v) \mid v \in D_{\bb(f)} \right\} \\
&= \sup\left\{ \gamma(v) \mid v \in \bigcup_{w \in D_f} f[w] \right\} \\
&= \sup\left\{ \gamma(v) \mid v \in \bigcup_{w \in D_f} V(w) \right\} \\
\end{align*}

By the precondition, we know that:

\[
\forall y \in V(x) \ .\ \gamma(y) < \gamma(x) 
\]

If we take an arbitrary $v \in \bigcup_{w \in D_f} V(w)$, there exists a $w \in
D_f$, s.t. $v \in V(w)$. The precondition then implies, that $\gamma(v) <
\gamma(w)$.  Additionally, since $w \in D_f$, by the definition of $\alpha(f)$,
it follows that $\gamma(w) \le \alpha(f)$.  Since all sets in question are
finite, $D_{\bb(f)}$ is either empty, in which case $\alpha(\bb(f)) = -\infty$
and the lemma trivially holds, or the supremum is actually a maximum, and
strict inequality is maintained when we infer:

\[
\alpha(\bb(f)) = \sup\left\{ \gamma(v) \mid v \in \bigcup_{w \in D_f} V(w) \right\} < \alpha(f)
\]
\end{proof}

As a trivial corollary to this lemma, we observe that $\alpha(\bb(f)) =
\alpha(f)$ iff $\alpha(f) = -\infty$.

\recallthm{thm}{\thmBody}
\begin{proof}
We prove this by using induction over $\alpha(f)$.

Base case $\alpha(f) = -\infty$: This can only be true if $D_f$ is empty.
However, it is then vacuously true that:

\[
\forall x \in D_f \ .\ P(x)
\]

It follows from the corollary to Lemma \ref{lemma1} that $D_g$ is empty too. In
that case, any universally quantified statement over $D_g$ vacuously holds.

\emph{General case}: Let $\alpha(f) = N \ge 0$ and assume the theorem holds for
any $f'$, for which $\alpha(f') < N$.  We observe that $f' \coloneqq \bb(f)$
satisfies all of the requirements. By definition, we have:

\[
\bb(f)[y] \coloneqq V(y)
\]

The above holds true for every element in its domain, so the precondition of
the theorem is met. Additionally, by Lemma~\ref{lemma2}, we know that
$\alpha(f') < N$, so we can use the induction assumption to conclude that:

\[
\forall x \in D_{\mop(f')} \ .\ \mop(f')[x] = \op(x)
\]

If it is the case that $\forall x \in D_f \ .\ P(x)$, then $\mop(f)[x] = e$ for
all domain elements, since $D_f = D_{\mop(f)}$, by the same reasoning we used
in the case where $D_f$ was empty. Similarly, for each such $x$, $\op(x) = e$
by the definition of $\op$, since $P(x)$ holds.

Next, we select an arbitrary $x \in D_{\mop(f)}$. We look at the case split in $\mapg$:
\begin{enumerate}
\item If $x \in D_{\mop(f')}$, by the induction hypothesis, we know that $\mop(f')[x] = \op(x)$.
\item If $x \notin D_{\mop(f')}$, but $P(x)$ holds, we know $\op(x) = e$. Here, $\mapg(f, \mop(f'))$ trivially evaluates to $e$ as well.
\item Otherwise, it remains to be shown that the following holds true:

\[
G(x, F(f[x], \mop(f'))) = \op(x)
\]

Since we know $P(x)$ does not hold for this $x$, it follows that:
\[
\op(x) = G(x, F(V(x), Op))
\]

Hence, it suffices to see that:

\[
F(f[x], \mop(f')) = F(V(x), Op)
\]

By the property of $f$, $f[x] = V(x)$. As $x \in D_f$, $f[x] \subseteq D_{f'}$
by definition and $D_{f'} \subseteq D_{\mop(f')}$ by Lemma~\ref{lemma1}, so for
every element $v \in V(x)$ it is the case that $\op(v) = \mop(f')[v]$. Since
$F$ is either a map or a filter, the equality above follows.

\end{enumerate}

\end{proof}

\recallcorollary{corollary}{\corollaryBody}
\begin{proof}
\[
f \coloneqq [ v \in \{x\} \mapsto V(x) ]
\]
obviously satisfies the precondition of Theorem \ref{thm}:
\[
\forall y \in D_f \ .\ f[y] = V(y)
\]
Since $x \in D_f \subseteq D_{\mop(f)}$ by Lemma~\ref{lemma1}, we can use Theorem~\ref{thm}, to conclude that $\mop(f)[x] = \op(x)$.
\end{proof}

\end{document}